\documentclass[pdflatex,sn-basic]{sn-jnl}

\usepackage{graphicx}%
\usepackage{multirow}%
\usepackage{amsmath,amssymb,amsfonts}%
\usepackage{amsthm}%
\usepackage{mathrsfs}%
\usepackage[title]{appendix}%
\usepackage{xcolor}%
\usepackage{textcomp}%
\usepackage{manyfoot}%
\usepackage{booktabs}%
\usepackage{algorithm}%
\usepackage{algorithmicx}%
\usepackage{algpseudocode}%
\usepackage{listings}%
\usepackage{natbib}          

\usepackage[T1]{fontenc}
\usepackage{lmodern}

\theoremstyle{plain}%
\newtheorem{theorem}{Theorem}
\newtheorem{proposition}{Proposition}%
\newtheorem{corollary}{Corollary}%
\newtheorem{assumption}{Assumption}%

\newtheorem*{remark}{Remark}%

\newtheorem{definition}{Definition}%

\begin{document}

\title[Article Title]{Identifying the potential of sample overlap in evidence synthesis of observational studies}

\author*[1]{\fnm{Zhentian} \sur{Zhang}}\email{zhentian.zhang@med.uni-goettingen.de}

\author[1]{\fnm{Tim} \sur{Friede}}\email{tim.friede@med.uni-goettingen.de}

\author[1]{\fnm{Tim} \sur{Mathes}}\email{tim.mathes@med.uni-goettingen.de}

\affil*[1]{\orgdiv{Department of Medical Statistics}, \orgname{University Medical Center Göttingen}, \orgaddress{\state{Lower Saxony}, \country{Germany}}}


\abstract{Sample overlap is a common issue in evidence synthesis in the field of medical research, particularly when integrating findings from observational studies utilizing existing databases such as registries. Due to the general inaccessibility of unique identifiers for each observation, addressing sample overlap has been a complex problem, potentially biasing evidence synthesis outcomes and undermining their credibility.
We developed a method to construct indicators for the degree of sample overlap in evidence synthesis of studies based on existing data. Our method is rooted in set theory and is based on the coding of the ranges of several well selected sample characteristics, offers a practical solution by focusing on making inference based on sample characteristics rather than on individual participant data. Useful information, such as the overlap-free sample set with the largest sample size in an evidence synthesis, can be derived from this method.
We applied our model to several real-world evidence syntheses, demonstrating its effectiveness and flexibility. Our findings highlight the growing importance of addressing sample overlap in evidence synthesis, especially with the increasing relevance of secondary use of data, an area currently under-explored in research.
}




\maketitle

\section{Introduction}\label{sec1}

Due to the rapid advancements in information technology and the expansion of data infrastructures in recent years, the medical field has witnessed an unprecedented increase in the volume of data produced and stored. However, a significant proportion of this data is redundant. Specifically, personal medical information, biomedical data, and clinical trial observations are often documented in multiple repositories, servers, and registers.

Observational studies that rely on existing data, such as clinical registries, represent important resources for answering epidemiological and clinical questions. Evidence synthesis is prone to bias due to sample overlap, especially when multiple included studies are based on similar databases\citep{Mathes2023}. Given the escalating volume of electronically collected and stored data, there is an anticipated surge in observational studies, which increases the risk of incorporating overlapping study samples into the same evidence synthesis. For simplicity, in this article we focus on registries, but the problem also exists for other types of electronic databases.

Overlap between samples can have a significant impact on the results of evidence synthesis. For example, \cite{Hussein2022} showed that in a meta-analysis, overlap can cause spurious high precision and considerable changes to the overall relative effect estimates. Despite the potentially critical impact of overlap, the awareness of this issue is often missing. 
Although there are some overviews of the overlap problem \citep{Mathes2023}, the lack of a systematic approach to handle this problem may further hinder researchers in addressing it.

Several attempts to address similar problems statistically have been made in the area of genome research, where secondary use of samples is common practice \citep{Lin2009, Han2016, Jin2020}. \cite{Lin2009} introduced a framework for the meta-analysis of GWAS (Genome-wide association studies) meta-analysis that can handle overlapping subjects between studies, which is simple to implement and computationally feasible for large GWAS. \cite{Han2016} proposed a method that transforms the covariance structure of the data, which could be then used in downstream analyses, and demonstrated the flexibility of the method using empirical datasets. \cite{Jin2020} proposed a method that resolves overlapping issue in gene-environment testing, utilizing Lin and Han’s correlation structures to generalize covariance matrices from the regular meta-regression model, and provided statistical tests for the joint effects of the gene main effect. However, there is very limited knowledge on the performance of their methods due to lack of comparison and application to (comparative) observational studies in other health care areas. In addition, in all of the three studies mentioned above, the overlapping portion in the samples of different genomics databases are assumed to be known in these approaches, whereas in fields other than genomics it might be difficult to estimate the true proportion of sample overlap. Efforts were also made in more general scenarios \citep{Bom2020, Lunny2021, Wolery2010}. Bom and Rachingers’ approach uses a generalized weight solution to handle sample overlap. By approximating the variance-covariance matrix that describes the correlation structure between outcomes, the method requires again the number of overlapping subjects as in the GWAS studies, which is rarely available outside of GWAS field. Also, their approach deviates from the common way of inverse variance weighting, is thus less easy to interpret. \cite{Lunny2021} considered overlapping of study instead of observations between studies, which is a related but different problem because the overlapping unit is the study. \cite{Wolery2010} compared four overlap methods for single-subject data synthesizing with visual analysts’ judgments, which only concentrates on the special case of single subject design (for example, the overlapping percentage of before and after treatment observations of a single patient). To our knowledge no established method and guidance exist on how to identify, quantify and account for overlap of observations (e.g. patients) systematically in evidence synthesis.

The focus of this paper is to provide a first step in addressing the sample overlap problem in evidence synthesis in the field of medical research. We will develop a theory to describe overlap, design methods and algorithms to derive useful information, such as the overlap-free sample sets with the largest sample size and the graph of the potential of overlap, and apply the methods in real-world evidence synthesis to show their practical value. Although the method is designed to improve systematic reviews and meta analyses of observational studies in medical research, the idea can be relevant in many other fields.





\section{Theoretical approach to identify sample overlap in evidence synthesis}\label{sec:dos}
\subsection{Preliminaries}
\subsubsection{Overlapping data vs longitudinal/clustered data}
The overlap considered in this paper refers to the same observations being included in multiple studies, which causes problems when we want to combine their results. This differs from clustered data structures, such as repeated measurements from the same patient in longitudinal settings or observations within the same unit in hierarchical models. In those cases, the data points correspond to distinct observations and therefore contribute additional information even if the outcomes are identical, albeit with dependencies that must be modeled. By contrast, in the overlapping case, repeated inclusion of records of the same observation does not increase the amount of information as the underlying events are identical.

\subsubsection{Overlap as a multivariate-relationship}
Overlap is a multivariate relationship. In contrast to bivariate relationship such as correlation, it cannot be fully represented using a single matrix. Figure~\ref{fig:triple overlap} shows an example of that.
\begin{figure}[ht]
    \centering
    \includegraphics[width=1\textwidth]{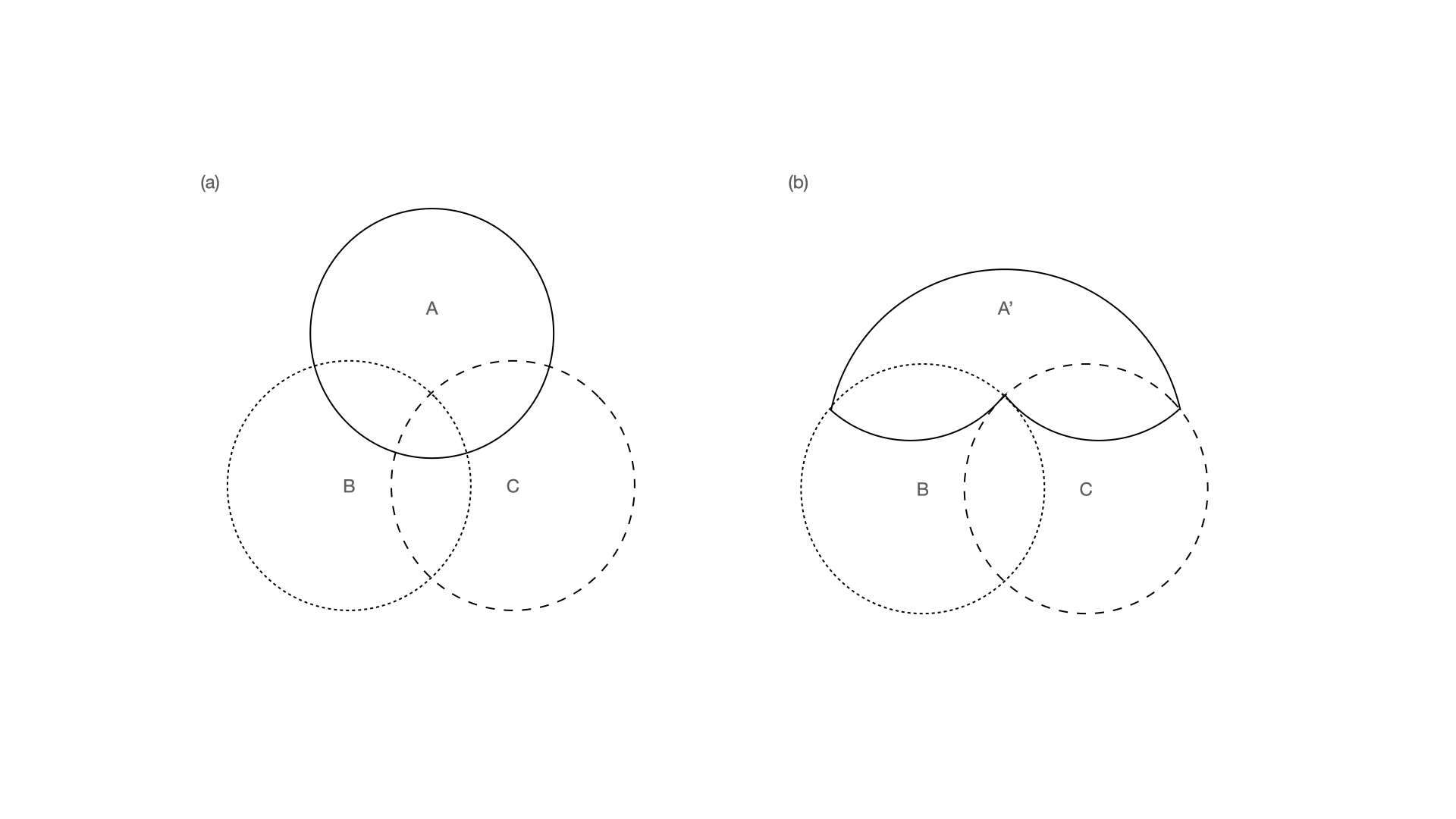}
    \caption{Pairwise overlap can be insufficient to characterize multivariate overlap. Two configurations are shown in which the pairwise intersections $(A\cap B, A\cap C, B\cap C)$ are identical, but the three-way intersection differs: $A\cap B\cap C\neq\emptyset$ in the left configuration, whereas $A'\cap B\cap C=\emptyset$ in the right. This illustrates that overlap among multiple studies cannot, in general, be recovered from pairwise overlap information alone.}

    \label{fig:triple overlap}
\end{figure}

For an evidence synthesis of $n$ study samples, $2^n-n-1$ combinations are of interest. $2^n$ is the cardinality of the power set of the set of samples, minus the number of studies, minus the empty set. In Figure~\ref{fig:combinations} we can see an example of all the combinations of an evidence synthesis of 6 study samples. Every column represents a combination, and the 57 combinations in the red rectangular are the ones whose overlapping situation we are interested in.
\begin{figure}[ht]
    \centering
    \includegraphics[width=1\textwidth]{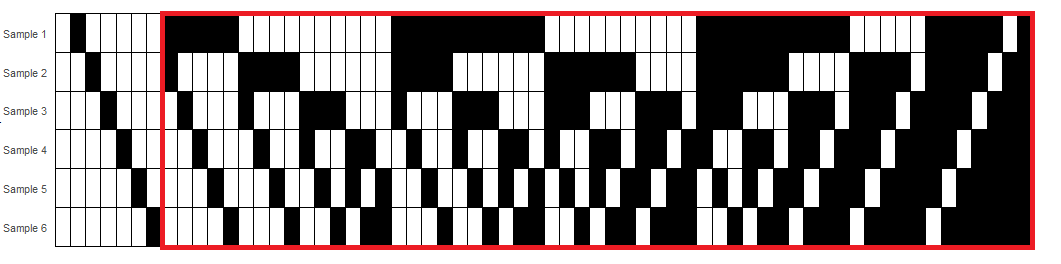}
    \caption{Study-combinations of $n=6$ studies. Columns correspond to subsets $A\subseteq\Omega$, rows correspond to studies $S_1,\ldots,S_6$, and a filled entry indicates membership $S_i\in A$. The highlighted block marks the $2^n-n-1$ non-empty combinations of size at least two for which overlap assessment is relevant.}
    \label{fig:combinations}
\end{figure}

\subsection{Formal setup}\label{sec: tmoo}
For simplicity and without loss of generality, we will develop our method in the context of meta-analysis using registry-based study in medical area, but the theory can be applied to evidence synthesis of observational studies in general. 

Registry studies utilize existing registry data to answer clinical questions. We use "registry" to represent any organized system that stores uniform data on a population \cite{Gliklich2020Registry}. In order to find the data suitable for a clinical question, the data are filtered using certain characteristics of the event (e.g. measure of outcome, time and location of observation, age of subjects). Such characteristics are the measured and stored versions of the underlying characteristics of the observation. Our approach uses information from these characteristics to model overlap of study samples when combining the study results.
\subsubsection{Notations}
\begin{definition}[Intrinsic characteristic vector]\label{def:icv}
Let $u$
uniquely index an underlying observation event (e.g., a particular subject at a particular measurement time/encounter).
The intrinsic characteristic vector of that observation event, given $n_k$ intrinsic characteristics, is a (latent) vector
\[
\mathbf{x}_u = (d_{u,1},\ldots,d_{u,n_k})
\]
collecting these characteristics inherent to that event.
By inherent and intrinsic we mean that the characteristics are well-defined for an event: for any $u,v$ and any $k\in\{1,\ldots,n_k\}$,
\[
u = v \;\Rightarrow\; d_{u,k} = d_{v,k}
\qquad\text{(equivalently, } d_{u,k}\neq d_{v,k}\Rightarrow u\neq v\text{)}.
\]
\end{definition}
The dimensions of $\mathbf{x}_u$ can be for example the underlying outcome value, the event time and location,
and subject/environmental characteristics at the time of observation, etc.
The exact values of certain dimensions of $\mathbf{x}_u$ may be unknown in practice.
Even for dimensions that are, in principle, measurable at arbitrarily high resolution (e.g., time or location),
real-world observations are recorded with finite precision and may be subject to measurement error or missingness.
In this section we stay with $\mathbf{x}_u$ for framing the basic setups, and show our approach for the real-world versions of it in the next sections.

Let $S_C=\{\mathbf{x}_1,\ldots,\mathbf{x}_{n_0}\}$ denote the collective sample, i.e., the whole set of distinct intrinsic characteristics vectors underlying all the existing recorded version of them in relevant registries for the target population.
For study $i\in\{1,\ldots,n\}$, let $S_i\subseteq S_C$ be its within-study overlap-free study sample of size $n_i$, written as
$S_i=\{\mathbf{x}_{I_{i,1}},\ldots,\mathbf{x}_{I_{i,n_i}}\}$ with indices
$I_{i,j}\in\{1,\ldots,n_0\}$ and $I_{i,j}\neq I_{i,\ell}$ for $j\neq \ell$.
Let $\Omega=\{S_1,\ldots,S_n\}$ denote the set of all study samples in an evidence synthesis.
For any $A\subseteq\Omega$, $\bigcup_{S_i\in A} S_i$ is then the pooled sample of the studies in $A$. Denote $S=\bigcup_{S_i\in \Omega} S_i$ as the sample of the meta-analysis. In other words, the meta-analysis sample is the overlap-free aggregation of the intrinsic characteristics vectors that are used by the studies included in the meta-analysis.
\begin{remark}
    By requiring $ I_{i, j}\neq I_{i, l} | j\neq l $ in the setup, we assume no overlap within each study sample.  In other words we regard $S_i$s as sets and thus could use the language of set theory to describe the overlap structure among them.
\end{remark}




From the setup, we have $S\subseteq S_C$. In the ideal case where there is no sample overlap within a meta-analysis, its $S$ is just the direct aggregation of $S_i$s. However, due to the overlap between $S_i$s, $|S|$ can be smaller than the sum of the individual sample sizes of the studies. To know how much smaller it is, we need definitions to describe the finer structure of overlap relationship between studies. This motivates following definition:
\begin{definition}[Overlap set and proportion of overlap]\label{def:ooasos}
For any $A\subseteq\Omega$ with $|A|\ge 2$, define the overlap set
\[
O(A):=\bigcap_{S_i\in A} S_i.
\]
and the proportion of overlap
$$\pi(A)= \frac{ \vert O(A) \vert}{ \vert \bigcup\limits_{S_i\in A} S_i \vert}$$
For $|A|<2$ we set $O(A):=\emptyset$ and $\pi(A):=0$ by convention.
\end{definition}

$O(A)$ provides a unambiguous basis for the discussion of overlap;
$\pi(A)$ gives a possible quantitative description of the degree of overlap among $A$.

\subsubsection{Overlap structure}\label{sec:os}
We defined $\pi(A)$ to quantify overlap. It is just one of the possible description of overlap.  We define overlap structure to generalize it to a range of such descriptions via real valued functions.

\begin{definition}[Overlap structure]\label{def:os}
Let $\Omega$ be a set of studies. An overlap structure on $\Omega$ is the set
\[\{(A, f(A)) : A \subseteq \Omega\},
\]
i.e., the graph of a set function $f:2^\Omega\to\mathbb{R}$ that assigns each study-combination
$A$ a real number $f(A)$ quantifying the overlap within $A$. 
\end{definition}

Examples for $f(A)$:
\[
f_1(A)=\mathbb{I}\{O(A)\neq\emptyset\},\qquad
f_2(A)=|O(A)|,\qquad
f_3(A)=\pi(A),
\qquad
f_4(A)=\frac{ \vert O(A) \vert}{ \left(\prod\limits_{S_i\in A} |S_i|\right)^{1/|A|}
}.
\]

In other words, for a given set A, $f_1(A)$ tells us if there is overlap or not in A, $f_2(A)$ gives the amount of overlapping elements and $f_3(A),f_4(A)$ gives different quantitative description of the degrees of overlap.

Example:\\
Let $\Omega=\{S_1,S_2,S_3,S_4\}$ ; $S=\{\mathbf{x}_a,\mathbf{x}_b,...,\mathbf{x}_i\}$ and \\
$S_1$ (sample 1) =\{$\mathbf{x}_a,\mathbf{x}_b,\mathbf{x}_c$\}\\
$S_2$ (sample 2) =\{$\mathbf{x}_c,\mathbf{x}_d,\mathbf{x}_e,\mathbf{x}_f$\}\\
$S_3$ (sample 3) =\{$\mathbf{x}_b,\mathbf{x}_f$\}\\
$S_4$ (sample 4) =\{$\mathbf{x}_d,\mathbf{x}_f,\mathbf{x}_g,\mathbf{x}_h,\mathbf{x}_i$\}\\

\begin{figure}[ht]
\centering
    \includegraphics[width=0.5\textwidth]{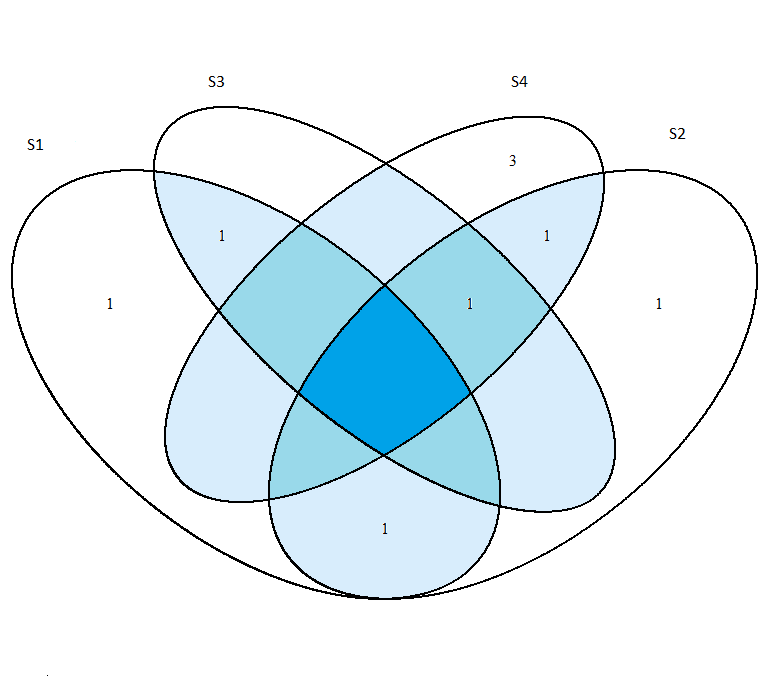}
   \caption{Example overlap structure for four study samples. The Venn diagram visualizes which observations are shared across $S_1,\ldots,S_4$. Shading intensity reflects multiplicity under naive sample-size addition (darker regions correspond to observations that would be counted more times if overlaps were ignored). This example underlies Table~\ref{tab:example_overlap_structure}, which reports $f(A)$ values for selected study-combinations.}

    \label{fig: Venn diagram with quantity}
\end{figure}

\begin{table}[ht]
\centering
\caption{Illustrative values of the overlap set $O(A)=\bigcap_{S_i\in A}S_i$ and four overlap summaries for selected study-combinations $A\subseteq\Omega$, where $\Omega=\{S_1,S_2,S_3,S_4\}$ is the toy example shown in Figure~\ref{fig: Venn diagram with quantity}. The functions are $f_1(A)=\mathbb{I}\{O(A)\neq\emptyset\}$, $f_2(A)=|O(A)|$, $f_3(A)=\pi(A)=|O(A)|\big/ \left|\bigcup_{S_i\in A}S_i\right|$, and $f_4(A)=|O(A)|\big/\left(\prod_{S_i\in A}|S_i|\right)^{1/|A|}$. For $|A|<2$ we use the conventions $O(A)=\emptyset$ and $f_1(A)=f_2(A)=f_3(A)=f_4(A)=0$.}

\label{tab:example_overlap_structure}
\begin{tabular}{|c|c|c|c|c|c|c|}
\hline
NR. & A & $O(A)$ & $f_1(A)$ & $f_2(A)$ & $f_3(A)$ & $f_4(A)$\\
\hline
1 & $\emptyset$ & $\emptyset$ & 0 & 0 & 0 & 0\\

2 & $\{S_1\}$ & $\emptyset$ & 0 & 0 & 0 & 0\\

3 & $\{S_2\}$ & $\emptyset$ & 0 & 0 & 0 & 0\\

4 & $\{S_3\}$ & $\emptyset$ & 0 & 0 & 0 & 0\\

5 & $\{S_4\}$ & $\emptyset$ & 0 & 0 & 0 & 0\\
\hline
6 & $\{S_1,S_2\}$ & \{$\mathbf{x}_c$\} & 1 & 1 & 1/6 $\approx$ 0.167 & 1/$2\sqrt{3}$ $\approx$ 0.289 \\

7 & $\{S_1,S_3\}$ & \{$\mathbf{x}_b$\} & 1 & 1 & 1/4 = 0.25& 1/$\sqrt{6}$ $\approx$ 0.408\\

8 & $\{S_1,S_4\}$ & $\emptyset$ & 0 & 0 & 0 & 0\\

9 & $\{S_2,S_3\}$ & \{$\mathbf{x}_f$\} & 1 & 1 & 1/5 = 0.2 & 1/$\sqrt{8}$ $\approx$ 0.354\\

10 & $\{S_2,S_4\}$ & \{$\mathbf{x}_d, \mathbf{x}_f$\} & 1 & 2 & 2/7 $\approx$ 0.286& 1/$\sqrt{5}$ $\approx$ 0.447\\

11 & $\{S_3,S_4\}$ & \{$\mathbf{x}_f$\} & 1 & 1 & 1/6 $\approx$ 0.167& 1/$\sqrt{10}$ $\approx$ 0.316\\

12 & $\{S_1,S_2,S_3\}$ & $\emptyset$ & 0 & 0 & 0 & 0\\

13 & $\{S_1,S_2,S_4\}$ & $\emptyset$ & 0 & 0 & 0 & 0\\

14 & $\{S_1,S_3,S_4\}$ & $\emptyset$ & 0 & 0  & 0 & 0\\

15 & $\{S_2,S_3,S_4\}$ & \{$\mathbf{x}_f$\} & 1 & 1  & 1/8 = 0.125& 1/2$\sqrt[3]{5}$ $\approx$ 0.292\\

16 & $\{S_1,S_2,S_3,S_4\}$ & $\emptyset$ & 0 & 0  & 0 & 0\\
\hline
\end{tabular}
\end{table}
\subsection{Identifying overlap structure using aggregated information}
The identification and quantification of overlap in evidence synthesis would be straightforward, if the intrinsic characteristic vectors are known: for each of the $2^n-n-1$ combinations of studies in an evidence synthesis, count the number of intrinsic characteristic vectors that appear in all studies in that combination. One of the problem that make the approach impractical is that the data on individual level are in general not available for meta-analyst, for example due to data security reasons. This problem can be partially addressed by using ranges of characteristics instead of single values of it. In the following we show how we use aggregated information to construct a proxy of the upper-bound for overlap structure. A more detail explanation of what we mean by aggregated data/information can be found in Appendix ~\ref{app:aggdata}.

Following previous notations, denote $d_{I_{i,j},k}$ as the value of the k-th dimension of $\mathbf{x}_{I_{i,j}}$, we have:
\begin{proposition}[Exclusion of pair-wise overlap by exclusion of the sets of one intrinsic characteristic]\label{prop:eopwobeosokc}
Denote $D_{i,k}:=\{d_{I_{i,j},k}\vert j\in \{1, 2,..., n_i\}\}$ as the set of the values of k-th dimension of the observations in $S_i$.\\
we have\\
$D_{i_1,k} \cap D_{i_2,k} = \emptyset \Rightarrow S_{i_1} \cap S_{i_2} = \emptyset$\\
\end{proposition}
The proof follows directly by applying the new notation to definition~\ref{def:icv}, which tells:
$(i_1, i_2, j_1, j_2),
d_{I_{i_1,j_1},k} \neq d_{I_{i_2,j_2},k} \Rightarrow I_{i_1,j_1} \neq I_{i_2,j_2}$. 

Proposition~\ref{prop:eopwobeosokc} is equivalent to: $\forall k, \bigcap\limits_{i\vert S_i\in A} D_{i,k} = \emptyset \Rightarrow\bigcap\limits_{i\vert S_i\in A} S_i = \emptyset$\\
\begin{remark}
    According to the Proposition~\ref{prop:eopwobeosokc}, if there are multiple intrinsic characteristics, the mutual exclusiveness of the set of the value of any one of the intrinsic characteristics would be enough to exclude the existence of overlap.\\
\end{remark}
\begin{remark}
Mention that the overlap of observations leads to the overlap between all the sets of the intrinsic characteristics, but not vice versa.\\
\end{remark}

\paragraph{From individual values to reported envelopes.}
The set $D_{i,k}$ is the set of individual-level values of characteristic $k$ among the observations used in study $i$.
In a typical evidence synthesis, individual values are not available to the meta-analyst. What is usually available instead are
study-level restrictions reported in the paper or study protocol (e.g., age $18$--$65$, calendar years $2010$--$2019$, ICD codes in a given list, hospitals in a given region). More discussion of it can be found in Appendix \ref{app:aggdata}.

We represent the theoretical counterpart of such reported information by a set $R_{i,k}$ that contain all individual values:
\[
D_{i,k}\subseteq R_{i,k}.
\]
In other words, $R_{i,k}$ is a description of where those values are allowed to lie according to what is reported.

\begin{proposition}[Exclusion of overlap by exclusion of the ranges of intrinsic characteristic]\label{prop:eoobeotrokc}
Denote $R_{i,k} \supseteq D_{i,k} $ as the ranges of $d_{I_{i,j},k}$s. $\forall A \subseteq \Omega$, we have\\
$\exists k, \bigcap\limits_{i\vert S_i\in A} R_{i,k} = \emptyset \Rightarrow \bigcap\limits_{i\vert S_i\in A} S_i = \emptyset$
\end{proposition}
Proof of it can be found in Appendix~\ref{sec:p}.\\
\subsubsection{Potential of overlap}\label{sec:roo}

Proposition~\ref{prop:eoobeotrokc} allows us to make inferences about the qualitative overlap structure among the studies with only ranges of intrinsic characteristics. 

Intuitively, the size of $\bigcap\limits_{i\vert S_i\in A} R_{i,k}$ also contains information about overlap: the smaller the common part of the reported ranges, the lesser amount of overlap is expected; a larger common part leaves more room for overlap.
However, in practice, “size” is hard to define in a uniform way across different kinds of characteristics (continuous time, discrete age groups, ICD codes, regions) and across different reporting formats (intervals, sets, unions of intervals). To obtain a representation that is simple and comparable across characteristics, we map ranges to binary vectors after partitioning the entire range into bins.

More specifically, we code $R_{i,k}$ into binary vectors through following steps:\\
For each dimension $k$,\\
1. Calculate the entire range of all possible value of k-th intrinsic characteristic $R_{\cdot,k}= \bigcup_{i=1}^n R_{i,k}$\\
2. Divide $R_{\cdot,k}$ into $m_k$ ordered and mutually disjoint subsets such that the union of these subsets equals $R_{\cdot,k}$. Denote the l-th subset as $R_{\cdot,k,l}, l\in \{1, 2,..., m_k\}$\\
3. For each study $i$, encode the range $R_{i,k}$ as a binary vector $\mathbf r_{i,k}\in\{0,1\}^{m_k}$ with entries
\[
r_{i,k,l}=\mathbb{I}\!\left\{\,R_{i,k}\cap R_{\cdot,k,l}\neq\emptyset\,\right\},\qquad l=1,\ldots,m_k,
\]
so that $r_{i,k,l}=1$ iff $R_{i,k}$ overlaps the $l$-th reference interval $R_{\cdot,k,l}$.

\begin{proposition}[Excluding pair-wise overlap based on the range vector of intrinsic characteristics]\label{prop:epobotrvokc} 
$\forall i_1, i_2, \\
\mathbf{r}_{i_1,k} \cdot \mathbf{r}_{i_2,k} = 0
\Rightarrow S_{i_1} \cap S_{i_2} = \emptyset$
\end{proposition}
Proof of it can be found in Appendix~\ref{sec:p}.\\

\begin{theorem}[Exclusion of overlapping sample combination]\label{thm:eoosp}
$\forall A \subseteq \Omega$ and for all intrinsic characteristics $k$, we have:\\
$\prod\limits_{k}(\overset{m_k}{\underset{l=1}{\sum}}\prod\limits_{i|S_i \in A} r_{i,k,l})=0 \Rightarrow \bigcap\limits_{S_i\in A} S_i = \emptyset$\\
\end{theorem}

This theorem follows directly from Proposition~\ref{prop:eoobeotrokc}, only substitutes $R_{i,k}$s with the partitioned and coded versions of $R_{i,k}$s .

\begin{definition}[Potential of overlap]\label{roo}
    
    $\forall A \subseteq \Omega$ and for all intrinsic characteristics $k$, given a family of partitions $\mathcal{P}=(\mathcal{P}_1,\ldots,\mathcal{P}_{n_k})$ of
$(R_{\cdot,1},\ldots,R_{\cdot,n_k})$, where $\mathcal{P}_k=(R_{\cdot,k,1},\ldots,R_{\cdot,k,m_k})$. We define the potential of overlap of $A$ :\\
    $\tilde{\pi}_{\mathcal{P}}(A):=\underset{k}{\operatorname{min}}\frac{ \overset{m_k}{\underset{l=1}{\sum}}\prod\limits_{i|S_i \in A} r_{i,k,l}}{m_k-\overset{m_k}{\underset{l=1}{\sum}}\prod\limits_{i|S_i \in A} (1-r_{i,k,l})}$\\
\end{definition}



\paragraph{Interpretation of the potential of overlap.}
Given a partition family $\mathcal P=(\mathcal P_1,\ldots,\mathcal P_{n_k})$ with
$\mathcal P_k=(R_{\cdot,k,1},\ldots,R_{\cdot,k,m_k})$.
Think of the bins $R_{\cdot,k,1},\ldots,R_{\cdot,k,m_k}$ as a coarse grid for the entire range of the
$k$th intrinsic characteristic.

For each study $i$ and characteristic $k$, the reported range $R_{i,k}$ is encoded as a binary vector
$\mathbf r_{i,k}\in\{0,1\}^{m_k}$ where $r_{i,k,l}=1$ if the range $R_{i,k}$ covers the $l$th bin,
and $r_{i,k,l}=0$ otherwise. In other words, $r_{i,k,l}$ simply records whether the $l$th bin is
included in the study's range for that characteristic.

For a study-combination $A\subseteq\Omega$ with $|A|\ge 2$ and a fixed characteristic $k$, define
\[
\tilde{\pi}_{\mathcal P,k}(A)
:=
\frac{ \overset{m_k}{\underset{l=1}{\sum}}\prod\limits_{i|S_i \in A} r_{i,k,l}}{m_k-\overset{m_k}{\underset{l=1}{\sum}}\prod\limits_{i|S_i \in A} (1-r_{i,k,l})}
\in[0,1].
\]
This quantity is easy to read:
\begin{itemize}
\item The numerator $\overset{m_k}{\underset{l=1}{\sum}}\prod\limits_{i|S_i \in A} r_{i,k,l}$ counts the number of bins that are
included by all studies in $A$ (i.e., bins $l$ for which $r_{i,k,l}=1$ for every $S_i\in A$).
\item The denominator $m_k-\overset{m_k}{\underset{l=1}{\sum}}\prod\limits_{i|S_i \in A} (1-r_{i,k,l})$ counts the number of bins
that are included by at least one study in $A$ (i.e., bins $l$ for which $r_{i,k,l}=1$ for
some $S_i\in A$).
\end{itemize}
So $\tilde{\pi}_{\mathcal P,k}(A)$ is the proportion of the ``combined range'' (measured in bins)
that is shared by all studies in $A$.

We then define the \emph{potential of overlap}
\[
\tilde{\pi}_{\mathcal P}(A):=\min_{k=1,\ldots,n_k}\tilde{\pi}_{\mathcal P,k}(A).
\]
We take the minimum because truly shared observations would have to fall into the shared part of the
ranges for every characteristic. If even one characteristic has very little shared range (in
bins), that alone already limits how much overlap between the studies is plausible given the
aggregated information.

\paragraph{Why this is a sensible proxy.}
The only information used here is which bins are covered by each study's reported range.
If, for some characteristic $k$, there is no bin that is covered by all studies in $A$, then
there is no way for an observation to satisfy all studies' range restrictions on that characteristic,
and overlap is ruled out for $A$ (cf.\ Theorem~\ref{thm:eoosp}).
When shared bins do exist, $\tilde{\pi}_{\mathcal P,k}(A)$ summarizes how large the shared part of the
ranges is compared with the overall part of the ranges, in a scale-free way.

\paragraph{Basic properties.}
For any fixed partition family $\mathcal P$:
\begin{itemize}
\item (Range) $\tilde{\pi}_{\mathcal P}(A)\in[0,1]$ for all $A$ with $|A|\ge 2$.
\item (Exclusion) If there exists $k$ such that
$ \overset{m_k}{\underset{l=1}{\sum}}\prod\limits_{i|S_i \in A} r_{i,k,l}=0$, then $\tilde{\pi}_{\mathcal P}(A)=0$,
which is consistent with Theorem~\ref{thm:eoosp}.
\item (Monotonicity in $A$) For a given $k$, when adding more studies to $A$, the number of bins shared by
all studies cannot increase, while the number of bins covered by at least one study cannot decrease.
Hence $\tilde{\pi}_{\mathcal P,k}(A)$ is non-increasing as $A$ grows, and therefore
$\tilde{\pi}_{\mathcal P}(A)$ is also non-increasing in $A$.
\item (Adding characteristics) Adding an additional intrinsic characteristic (i.e., taking the
minimum over more $k$) cannot increase $\tilde{\pi}_{\mathcal P}(A)$.
\item (Dependence on the partition) Using a finer partition typically separates values that were
previously merged into the same bin. This often reduces artificial sharing caused by coarse bins, so
$\tilde{\pi}_{\mathcal P}(A)$ often becomes smaller under finer partitions, although strict
monotonicity is not guaranteed.
\end{itemize}

\paragraph{Conservativeness and limitations.}
$\tilde{\pi}_{\mathcal P}(A)$ is a deterministic function of the reported ranges and the chosen
partitioning scheme. Because binning is a coarsening step, different values that would be distinct on
the individual level may fall into the same bin; this can make the shared part (in bins) look larger
than it truly is. However, without additional assumptions linking individual-level distributions to
the reported ranges, a general inequality such as $\tilde{\pi}_{\mathcal P}(A)\ge \pi(A)$ cannot be
guaranteed. We therefore interpret $\tilde{\pi}_{\mathcal P}(A)$ as a feasibility-based measure of
overlap potential rather than an estimator of the real overlap proportion.

\section{Overlap identification in the real world}
\subsection{Intrinsic VS Key characteristics}
In the theoretical part, we derived a proxy of the upper bound of the portions of overlap, using range of the intrinsic characteristics. Yet intrinsic characteristics and their ranges are theoretical subjects whose real values are unknown in the real world. The individual data used in primary studies are not intrinsic characteristic vectors, but a transformed version of it, often extracted from the so called IPD (Individual Participant Data) in medical research. The transformations happens at steps such as measuring, recording, extracting and cleaning, possibly with losses of resolution (e.g., rounding or categorization).

Once $\mathbf{x}_u$ is defined, each dimension of it has a unique, invariant value. By contrast, different transformed versions of the same intrinsic characteristic may differ in precision (e.g., the number of stored digits), in coarsening rules (e.g., rounding schemes), or in representation (e.g., 12/24-hour clocks or different time zones). They may also differ in how categorical information is partitioned and coded, for instance by using different levels of granularity (e.g., hospital $\to$ postcode $\to$ city $\to$ province $\to$ country).

Motivated by this, we distinguish between the intrinsic characteristic vector of an observation and its transformed version that is used for the analysis of it in a study. For a clear separation between the theoretical value of a variable and the transformed (i.e., observed/recorded/coarsened) version of it in the reality, we denote the transformed value of a variable by appending a prime~($'$).\\

\begin{definition}
Let $u$ uniquely index an underlying observation event. For ${I_{i,j}}=u$, we define $\mathbf{x}'_{I_{i,j}}=(d'_{I_{i,j},k_1}, d'_{I_{i,j},k_2}, ..., d'_{I_{i,j},k_{n'_k}})$ as the key characteristic vector of the event indexed by $u$ in study $i$, where $k_1,\ldots,k_{n'_k}\in{1,\ldots,n_k}$ are distinct indices identifying which intrinsic characteristics are used as key characteristics in study $i$. $d'_{I_{i,j},k}$ are transformed versions of intrinsic characteristics $d_{u,k}$ that are used for the analysis in the observational study $i$. 
\end{definition}
\begin{remark}
Characteristics representing administrative data, database-related or study-related information (e.g., registry identifier, storage timestamp, current insurance status of the patient, author names and publication year etc.) are \textbf{not} key characteristics, because they are not intrinsic to the underlying event and could vary between studies. Thus their cannot be used for overlap inference in our approach.
\end{remark}

\subsection{Potential of overlap based on key characteristics.}
In practice, the meta-analyst observes only study-level restrictions on key characteristics.
For each study $i$ and characteristic $k$, let $R'_{i,k}$ denote the reported range of the key characteristic
(e.g., age in years, calendar time in months, ICD code set at a chosen level), and let
$R'_{\cdot,k}:=\bigcup_{i=1}^n R'_{i,k}$ be the overall reported range across studies.
Choose a partition family $\mathcal P'=(\mathcal P'_1,\ldots,\mathcal P'_{n'_k})$ with
$\mathcal P'_k=(R'_{\cdot,k,1},\ldots,R'_{\cdot,k,m'_k})$ and encode
\[
r'_{i,k,l}=\mathbb I\{R'_{i,k}\cap R'_{\cdot,k,l}\neq\emptyset\},\qquad l=1,\ldots,m'_k.
\]
Define the key-characteristic analogue of the potential of overlap by
\[
\tilde{\pi}'_{\mathcal P'}(A)
:= \underset{k}{\operatorname{min}}\frac{ \overset{m'_k}{\underset{l=1}{\sum}}\prod\limits_{i|S_i \in A} r'_{i,k,l}}{m'_k-\overset{m'_k}{\underset{l=1}{\sum}}\prod\limits_{i|S_i \in A} (1-r'_{i,k,l})}
\]

By assuming that the distortion of the mapping from $d_{I_{i,j},k}$ to $d'_{I_{i,j},k}$ is ignorable given our partition of the aggregated information, we can bridge the gap between theory and practice.


\begin{assumption}[Partition-compatibility of the transformation]\label{ass:partition_compatibility}
For every included observation $u=I_{i,j}$, every characteristic $k$, and every bin $l$,
\[
d'_{I_{i,j},k}\in R'_{\cdot,k,l}\ \Longrightarrow\ d_{u,k}\in R'_{\cdot,k,l}.
\]
Equivalently, relative to the chosen partition $\mathcal P'_k$, the transformation from $d_{u,k}$ to $d'_{I_{i,j},k}$
does not move values across bin boundaries. 
\end{assumption}

\begin{corollary}[Real-world overlap exclusion and properties]\label{cor:real_world_potential}
Under Assumption~\ref{ass:partition_compatibility}, the exclusion statements derived from
Theorem~\ref{thm:eoosp} remain valid when replacing $(R_{i,k},r_{i,k,l})$ by $(R'_{i,k},r'_{i,k,l})$.
In particular, $\tilde{\pi}'_{\mathcal P'}(A)=0$ implies $\bigcap_{S_i\in A}S_i=\emptyset$.

Moreover, the basic properties of the potential of overlap (range in $[0,1]$, monotonicity in $A$,
and non-increase when adding characteristics) hold for $\tilde{\pi}'_{\mathcal P'}(A)$ exactly as before,
since they depend only on the binary encoding.
\end{corollary}
\subsection{Example}\label{sec:example_real_world}

This toy example illustrates following points:
(i) if linked individual-level extracted data were available to the meta-analyst, one could in principle identify the overlap exactly;
(ii) across studies, the same intrinsic characteristic (e.g., time) may appear in the extracted data at different granularity (year vs month vs timestamp with time zone);
(iii) in typical evidence synthesis, only study-level envelopes of a few key characteristics are available, and these can be used to compute the overlap potential \(\tilde{\pi}'_{\mathcal P'}\).

\paragraph{A hypothetical individual-level view (not usually available).}
Table~\ref{tab:example_4_studies} shows what the extracted individual-level data could look like in a hypothetical setting where the meta-analyst has access to linked records across studies.
Each row corresponds to an included observation event \(u\), but the recorded variables are study-specific versions \(d'_{I_{i,j},k}\), not the real values \(d_{u,k}\).
Crucially, the extracted representation can differ between studies even for the same intrinsic characteristic:
for example, time is recorded as
a full timestamp with time zone (Study 3),
a month-year value (Study 2),
or only the year (Study 4).

\begin{table}[ht]
\centering
\caption{A hypothetical example of extracted individual-level data (not only key characteristics!) for the four studies.
This table is shown for illustration only; meta-analysts typically do not have access to such individual-level records.}
\label{tab:example_4_studies}
\begin{tabular}{|c|c|c|c|c|c|c|}
\hline
Study & ID & outcome & location of observation& time of observation & group & publication year of paper\\
\hline
1 & a & 1 & area 1 & Aug-18-2021 & Treatment & 2023\\
1 & b & 0 & area 2 & Dec-03-2022 & Control & 2023\\
1 & c & 0 & area 2 & Nov-13-2022 & Control & 2023\\
\hline
2 & d & 1 & area 3 & Feb-2023 & Treatment & 2024\\
2 & c & 0 & area 2 & Nov-2022 & Control & 2024\\
2 & e & 0 & area 1 & Oct-2023 & Control & 2024\\
2 & f & 1 & area 3 & Apr-2023 & Treatment & 2024\\
\hline
3 & f & 1 & area 3 & Apr-18-2023, 4:40 CET & Treatment & 2023\\
3 & b & 0 & area 2 & Dec-03-2022, 5:20 CET & Control & 2023\\
\hline
4 & f & 1 & area 3 & 2023 & Treatment & 2024\\
4 & g & 1 & area 3 & 2023 & Treatment & 2024\\
4 & d & 1 & area 3 & 2023 & Treatment & 2024\\
4 & h & 1 & area 4 & 2023 & Control & 2024\\
4 & i & 0 & area 4 & 2023  & Control & 2024\\
\hline
\end{tabular}
\end{table}

In this constructed example, the IDs \(b,c,d,f\) appear in multiple studies, so the realized overlaps (on the latent event level) are non-zero for some combinations.
For instance, \(S_1\) and \(S_4\) share no events, while \(S_2\) and \(S_4\) share two events (IDs \(d\) and \(f\)).
Of course, in real applications such cross-study linkage is usually not available to the meta-analyst; we include it here only to define the ``ground truth'' overlap proportion \(\pi(A)\) for comparison.

\paragraph{Which characteristics are useful for overlap inference via \texorpdfstring{\(\tilde{\pi}'_{\mathcal P'}\)}{pi-tilde-prime}?}
Not every column in Table~\ref{tab:example_4_studies} is suitable for overlap inference.
A characteristic is not useful for our overlap-potential approach if it is
(a) non-intrinsic (can differ between studies even for the same event), or
(b) not reported in a way that yields comparable study-level envelopes across studies.

For example, ``publication year'' (and also author list, registry identifier, extraction timestamp, etc.) is non-intrinsic and therefore excluded.
Outcome is intrinsic, but it is often not reported as a restriction that defines the eligible sample, and for many problems it provides little exclusion power (e.g., binary outcomes).
In contrast, time and location are common intrinsic characteristics and are often reported as sample restrictions (eligibility windows, regions), making them natural key characteristics for \(\tilde{\pi}'_{\mathcal P'}\).

\paragraph{What is typically available: study-level envelopes of key characteristics.}
In most evidence syntheses, the meta-analyst can extract only aggregated information from the included studies.
For this example, Table~\ref{tab:ad} summarizes the kind of information that is usually accessible: sample size,
and (crucially) reported envelopes for location and time.

\begin{table}[ht]
\centering
\caption{A hypothetical example of aggregated data (constructed from Table~\ref{tab:example_4_studies}) that is typically accessible to meta-analysts. '\#' means 'the number of'.}
\label{tab:ad}
\renewcommand{\arraystretch}{1.2}
\setlength{\tabcolsep}{4pt}
\begin{tabular}{|c|c|c|c|c|c|c|c|}
\hline
sample & \begin{tabular}[c]{@{}c@{}}sample \\ size\end{tabular} & \begin{tabular}[c]{@{}c@{}}range of \\ location\end{tabular} & \begin{tabular}[c]{@{}c@{}}range of \\ time\end{tabular} & \begin{tabular}[c]{@{}c@{}}\# subjects: \\ control\end{tabular} & \begin{tabular}[c]{@{}c@{}}\# subjects: \\ treatment\end{tabular} & \begin{tabular}[c]{@{}c@{}}\# event: \\ control\end{tabular} & \begin{tabular}[c]{@{}c@{}}\# event: \\ treatment\end{tabular} \\
\hline
1 & 3 & \{area 1, area 2\} & \{2021, 2022\} & 2 & 1 & 0 & 1\\
\hline
2 & 4 & \{area 1, area 2, area 3\} & \{2022, 2023\} & 2 & 2 & 0 & 2\\
\hline
3 & 2 & \{area 2, area 3\} & \{2022, 2023\} & 1 & 1 & 0 & 1\\
\hline
4 & 5 & \{area 3, area 4\} & \{2023\} & 2 & 3 & 1 & 3\\
\hline
\end{tabular}%
\end{table}

Formally, we interpret the reported ``range of location'' and ``range of time'' as envelopes of key characteristics:
\[
R'_{i,1}=\text{reported location envelope},\qquad
R'_{i,2}=\text{reported time envelope (here: years)}.
\]
For example, \(R'_{1,1}=\{\text{area 1},\text{area 2}\}\) and \(R'_{4,1}=\{\text{area 3},\text{area 4}\}\), hence
\(R'_{1,1}\cap R'_{4,1}=\emptyset\), which, under Assumption~\ref{ass:partition_compatibility}, already excludes overlap between \(S_1\) and \(S_4\), and also excludes any combination containing both.

\paragraph{Encoding the envelopes and computing \texorpdfstring{\(\tilde{\pi}'_{\mathcal P'}\)}{pi-tilde-prime}.}
Here the overall reported ranges are
\[
R'_{\cdot,1}=\{\text{area 1},\text{area 2},\text{area 3},\text{area 4}\},\qquad
R'_{\cdot,2}=\{2021,2022,2023\}.
\]
Choose partitions \(\mathcal P'_1\) and \(\mathcal P'_2\) given by the singleton bins, and encode
\(r'_{i,k,l}=\mathbb{I}\{R'_{i,k}\cap R'_{\cdot,k,l}\neq\emptyset\}\).
The resulting binary vectors are:

\begin{table}[ht]
\centering
\caption{Binary encoding of the location envelopes $R'_{i,1}$ for the four-study toy example. The overall domain is partitioned into singleton bins $R'_{\cdot,1,l}$ (one per area). For each study $i$ and bin $l$, the entry equals $r'_{i,1,l}=\mathbb{I}\{R'_{i,1}\cap R'_{\cdot,1,l}\neq\emptyset\}$, i.e., it is $1$ if the study's reported location envelope includes that area and $0$ otherwise.}

\label{tab:coding_location_prime}
\begin{tabular}{lccccc}
\toprule
Location & Partition of \(R'_{\cdot,1}\)& \(\mathbf{r}'_{1,1}\) & \(\mathbf{r}'_{2,1}\) & \(\mathbf{r}'_{3,1}\) & \(\mathbf{r}'_{4,1}\) \\
\midrule
area 1       & \(R'_{\cdot,1,1}\) & 1  & 1  & 0  & 0  \\
area 2       & \(R'_{\cdot,1,2}\) & 1  & 1  & 1  & 0  \\
area 3       & \(R'_{\cdot,1,3}\) & 0  & 1  & 1  & 1  \\
area 4       & \(R'_{\cdot,1,4}\) & 0  & 0  & 0  & 1  \\
\bottomrule
\end{tabular}
\end{table}

\begin{table}[ht]
\centering
\caption{Binary encoding of the year envelopes \(R'_{i,2}\) for the four-study example.}
\label{tab:coding_year_prime}
\begin{tabular}{lccccc}
\toprule
Year & Partition of \(R'_{\cdot,2}\) & \(\mathbf{r}'_{1,2}\) & \(\mathbf{r}'_{2,2}\) & \(\mathbf{r}'_{3,2}\) & \(\mathbf{r}'_{4,2}\)\\
\midrule
2021 & \(R'_{\cdot,2,1}\) & 1  & 0  & 0  & 0  \\
2022 & \(R'_{\cdot,2,2}\) & 1  & 1  & 1  & 0  \\
2023 & \(R'_{\cdot,2,3}\) & 0  & 1  & 1  & 1  \\
\bottomrule
\end{tabular}
\end{table}

For \(A_1=\{S_1,S_2\}\),
\[
\tilde{\pi}'_{\mathcal P',1}(A_1)=\frac{2}{3}\quad\text{(shared areas: 1 and 2)},
\qquad
\tilde{\pi}'_{\mathcal P',2}(A_1)=\frac{1}{3}\quad\text{(shared years: only 2022)},
\]
so
\[
\tilde{\pi}'_{\mathcal P'}(A_1)=\min\left\{\frac{2}{3},\frac{1}{3}\right\}=\frac{1}{3}.
\]
In contrast, using the (hypothetical) linked individual-level IDs in Table~\ref{tab:example_4_studies},
the overlap proportion is
\[
\pi(A_1)=\frac{|O(A_1)|}{\left|\bigcup_{S_i\in A_1}S_i\right|}=\frac{1}{6},
\]
since only one event (ID \(c\)) is shared and the pooled sample size is 6.

\paragraph{Comparison across all combinations.}
Table~\ref{tab:roo_prime} compares the overlap potential \(\tilde{\pi}'_{\mathcal P'}(A)\), computed only from the envelopes,
to the overlap proportion \(\pi(A)\), known here only because we constructed a linked toy dataset.

\begin{table}[ht]
\centering
\caption{Comparison between overlap potential $\tilde{\pi}'_{\mathcal P'}(A)$ (computed from study-level envelopes of key characteristics using the chosen partition family $\mathcal P'$) and real overlap proportion $\pi(A)=|O(A)|\big/\left|\bigcup_{S_i\in A}S_i\right|$ in the four-study toy example. The overlap set $O(A)$ and $\pi(A)$ are shown only because the toy data include linked IDs (Table~\ref{tab:example_4_studies}); in real applications, meta-analysts typically observe only the envelopes used to compute $\tilde{\pi}'_{\mathcal P'}(A)$.}

\label{tab:roo_prime}
\begin{tabular}{|c|c|c|c|c|}
\hline
NR. & \(A\) & \(O(A)\) & \(\tilde{\pi}'_{\mathcal P'}(A)\) & \(\pi(A)\) \\
\hline
1 & \(\emptyset\) & \(\emptyset\) & / & / \\
2 & \(\{S_1\}\) & \(\emptyset\) & 0 & 0\\
3 & \(\{S_2\}\) & \(\emptyset\) & 0 & 0\\
4 & \(\{S_3\}\) & \(\emptyset\) & 0 & 0\\
5 & \(\{S_4\}\) & \(\emptyset\) & 0 & 0\\
\hline
6 & \(\{S_1,S_2\}\) & \(\{\mathbf{x}_c\}\) & \(1/3\) & \(1/6\)\\
7 & \(\{S_1,S_3\}\) & \(\{\mathbf{x}_b\}\) & \(1/3\) & \(1/4\)\\
8 & \(\{S_1,S_4\}\) & \(\emptyset\) & 0 & 0\\
9 & \(\{S_2,S_3\}\) & \(\{\mathbf{x}_f\}\) & \(2/3\) & \(1/5\)\\
10 & \(\{S_2,S_4\}\) & \(\{\mathbf{x}_d,\mathbf{x}_f\}\) & \(1/4\) & \(2/7\)\\
11 & \(\{S_3,S_4\}\) & \(\{\mathbf{x}_f\}\) & \(1/3\) & \(1/6\)\\
12 & \(\{S_1,S_2,S_3\}\) & \(\emptyset\) & \(1/3\) & 0\\
13 & \(\{S_1,S_2,S_4\}\) & \(\emptyset\) & 0 & 0\\
14 & \(\{S_1,S_3,S_4\}\) & \(\emptyset\) & 0 & 0\\
15 & \(\{S_2,S_3,S_4\}\) & \(\{\mathbf{x}_f\}\) & \(1/4\) & \(1/8\)\\
16 & \(\{S_1,S_2,S_3,S_4\}\) & \(\emptyset\) & 0 & 0\\
\hline
\end{tabular}
\end{table}

The role of \(\tilde{\pi}'_{\mathcal P'}(A)\) is thus best interpreted as a feasibility-based measure of overlap potential:
values equal to \(0\) provide strong exclusion evidence under Assumption ~\ref{ass:partition_compatibility}(Corollary~\ref{cor:real_world_potential}),
while positive values indicate that overlap is possible with the reported envelopes but not guaranteed.
Moreover, because the construction depends on the chosen partitions and on which key characteristics are available,
\(\tilde{\pi}'_{\mathcal P'}(A)\) should not be used as a universal upper bound on \(\pi(A)\).
\subsection{Visualization of overlap potential}\label{sec:viz}

The key-characteristic overlap potential
\(\tilde{\pi}'_{\mathcal P'}(A)\) is a set function on study-combinations \(A\subseteq\Omega\) (with \(|A|\ge2\)),
computed from reported envelopes \(R'_{i,k}\) after choosing a partition family \(\mathcal P'\).
Any visualization therefore is conditional on the selected key characteristics and on \(\mathcal P'\).

We suggest two complementary visualizations.

\paragraph{Pairwise heat map}
A heat map provides a compact overview of which study pairs have
large shared envelope mass (in bins) and which pairs are excluded.
This visualization is intuitive and often captures a substantial part of the practical overlap risk,
even though overlap is genuinely multivariate and cannot be fully summarized by pairwise information
(cf.\ Section~\ref{sec:os}).

For the four-study example (Table~\ref{tab:roo_prime}), the heat map
is shown in Figure~\ref{fig:4studyexampleheatmap_prime}.

\begin{figure}[!htbp]
\centering
\includegraphics[width=1\textwidth]{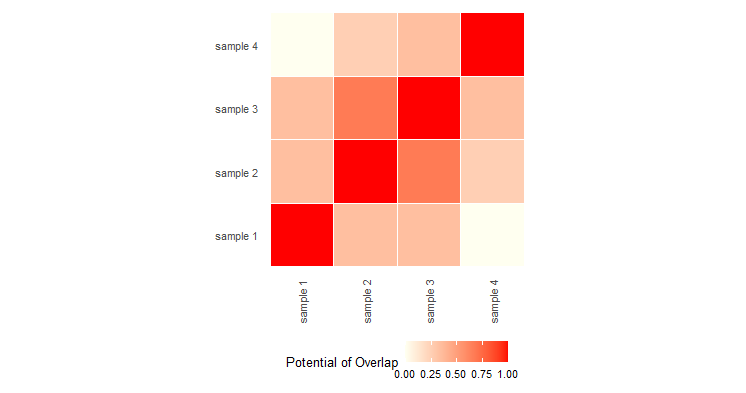}
\caption{Pairwise overlap-potential heat map for the four-study toy example. Cell $(i,j)$ shows $\tilde{\pi}'_{\mathcal P'}(\{S_i,S_j\})$, computed from study-level envelopes of the key characteristics (location and time) at the chosen partition resolution. Values equal to $0$ indicate overlap is excluded under Assumption~\ref{ass:partition_compatibility}; larger values indicate greater overlap compatibility given the reported envelopes.}

\label{fig:4studyexampleheatmap_prime}
\end{figure}

\paragraph{Combination grid plot (UpSet-plot-style)}
To visualize multivariate overlap structure, we plot \(\tilde{\pi}'_{\mathcal P'}(A)\) for selected non-empty combinations \(A\)
(e.g., all combinations with \(\tilde{\pi}'_{\mathcal P'}(A)>0\) in small examples, or the top \(K\) combinations in larger examples).
A practical visualization is a ``grid plot'' inspired by UpSet-plot suggested in \cite{Lex2014UpSet}. Each column corresponds to one combination \(A\),
rows correspond to studies \(S_1,\ldots,S_n\), and a filled cell indicates membership \(S_i\in A\).
Columns are ordered by decreasing \(\tilde{\pi}'_{\mathcal P'}(A)\), and combinations with \(\tilde{\pi}'_{\mathcal P'}(A)=0\)
may be omitted for conciseness.

In the four-study example, all non-zero combinations from Table~\ref{tab:roo_prime} are shown in
Figure~\ref{fig:4studyexample_grid_prime}. For larger \(n\), showing all \(2^n-n-1\) combinations is not meaningful;
in that case we recommend visualizing only the combinations with the largest \(\tilde{\pi}'_{\mathcal P'}(A)\).

\begin{figure}[!htbp]
\centering
\includegraphics[width=1\textwidth]{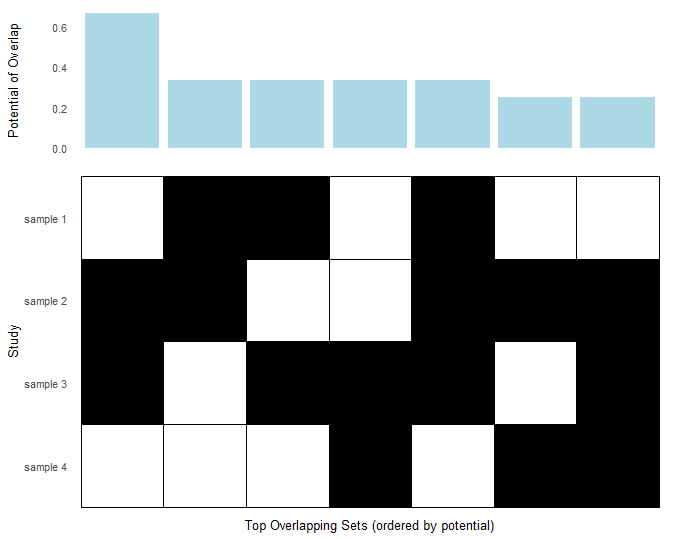}
\caption{Grid plot of non-zero overlap potentials in the four-study toy example (Table~\ref{tab:roo_prime}). Each column represents a study-combination $A\subseteq\Omega$; filled cells indicate which studies are included in $A$. Columns are ordered by decreasing $\tilde{\pi}'_{\mathcal P'}(A)$, and only combinations with $\tilde{\pi}'_{\mathcal P'}(A)>0$ are shown.}

\label{fig:4studyexample_grid_prime}
\end{figure}

\paragraph{Interpretation reminder.}
In both plots, \(\tilde{\pi}'_{\mathcal P'}(A)=0\) provides exclusion evidence under Assumption~\ref{ass:partition_compatibility}
(Corollary~\ref{cor:real_world_potential}).
Positive values indicate that overlap is possible with the reported envelopes (given \(\mathcal P'\)),
but do not imply actual overlap.

\section{Utilization of the potential of overlap}\label{sec:uotroo}

\subsection{Overlap-free sample combination with the maximum sample size}\label{sec:ofs}
Based on the potential of overlaps, we can derive subsets of $\Omega$ for which the aggregation of all the study samples within them are free of any overlaps. As the sample sizes of the studies are derivable from aggregated data, we can go a step further and choose the subset of study samples with the largest aggregated sample size to perform a new meta-analysis that is completely free of the impact of overlap.\\

More specifically, we conduct the following steps to find the set of studies:\\
1. We find the set $B_0$ of all the sets of study samples whose potential of overlaps are zero.\\
$$B_0:=\{A\vert A\in 2^\Omega, \tilde{\pi}(A)=0 \}$$
2. In $B_0$, we keep only the sets that have all their subsets also in $B_0$, and denote the new set as $B_1$.\\
3. In $B_1$, we keep only the sets who do not have a proper superset in $B_1$, and denote it as $B_2$.\\
4. Choose the set in $B_2$ that has the largest sample size. \\

\textbf{Example:}\\
From the 4-studies example, we list the result of the steps :\\
Step1: $B_0=\{\emptyset, \{S_1\}, \{S_2\}, \{S_3\}, \{S_4\},\{S_1, S_4\}, \{S_1, S_2, S_4\}, \{S_1, S_3, S_4\}, \{S_1, S_2, S_3, S_4\}\}$, as seen in Table~\ref{tab:roo_prime}\\
Step2: $B_1=\{\emptyset, \{S_1\}, \{S_2\}, \{S_3\}, \{S_4\},\{S_1, S_4\}\}$\\
Step3: $B_2=\{\{S_2\}, \{S_3\}, \{S_1, S_4\}\}$\\
Step4: $\{S_1, S_4\}$ is chosen because it has the largest sample size (3+5=8), as can be derived from aggregated data in Table~\ref{tab:ad}. We can add the sample size directly, because step 2 ensures that.\\
In our example, the original meta-analysis based on $\{\{S_1\}, \{S_2\}, \{S_3\}, \{S_4\}\}$ would have a spurious sample size of 3+4+2+5=14, and because of the considerable overlap, the actual meta-analysis sample size should be 9 instead. By applying the process we have a new overlap-free meta-analysis ($\{S_1, S_4\}$) of population size 8.\\

The meaning of calculating the overlap-free sample combination with the maximum sample size is that we can do meta-analysis based on that sample combination and compare with the original meta-analysis to evaluate the impact of overlap.\\
The considerable overlap between the study samples in our example makes the application of the process necessary and helpful. But in some cases, excluding the studies might mean throwing away too much information. We thus suggest an alternative way to look at the problem in Section~\ref{sec:eub}.\\

\subsubsection{Alternatives to sample-size maximization.}
When selecting a final overlap-free study combination from $B_2$ (Section~\ref{sec:ofs}), maximizing the total sample size is only one possible criterion. Depending on the goal of the evidence synthesis, other reasonable criteria include (but are not limited to) (i) minimizing between-study heterogeneity, (ii) maximizing the number of included studies, (iii) maximizing the total inverse-variance weight, and (iv) maximizing the incorporated risk-of-bias or reporting-quality scores. We could also use them in combination, for example by imposing minimum requirements (e.g., at least $m$ studies) and then selecting among the remaining candidates using a secondary criterion (e.g., lower heterogeneity), or by combining several criteria into a single score.

\textbf{Example:}\\
Suppose we have an evidence synthesis of 7 studies, where $$B_0^c:=\{A\vert A\in 2^\Omega, \tilde{\pi}'_{\mathcal P'}(A)\neq 0 \}=\{\{S_1, S_3\}, \{S_2, S_3\}, \{S_3, S_4\}, \{S_3, S_6\},\{S_4, S_7\}\}$$.\\
Step~1: $B_0$ consists of the sets marked with green in Figure~\ref{fig:7studyb0} in the appendices .\\
Step~2: $B_1$ consists of the sets marked with green in Figure~\ref{fig:7studyb1}.\\
Step~3: $B_2$ consists of the sets marked with green in Figure~\ref{fig:7studyb2}, which are $\{S_3, S_5, S_7\}$, $\{S_1, S_2, S_4, S_5, S_6\}$ and $\{S_1, S_2, S_5, S_6, S_7\}$\\
Step~4: We then evaluate these candidates using the chosen selection criteria. For instance, we can compute $\tau^2$ (or $I^2$) for each candidate and select the combination with lower estimated heterogeneity, while simultaneously enforcing other practical preferences (e.g., choosing among combinations with similar heterogeneity the one including more studies, or the one with larger total inverse-variance weight).

\subsection{Proxy of the lower bound of the meta-analysis sample size}\label{sec:eub}
When degrees of overlap is low between study samples, we might not want to simply exclude studies, as it could significantly shrink the meta-analysis sample size. An alternative is to construct a proxy for a lower bound for the effective meta-analysis sample size
\[
|S|=\left|\bigcup_{S_i\in\Omega} S_i\right|.
\]

\paragraph{A second-order bound from the principal of inclusion--exclusion.}
For any finite sets $\Omega = \{S_1,\ldots,S_n\}$, the principle of inclusion and exclusion gives
\[
\left|\bigcup_{S_i\in\Omega} S_i\right|
=
\sum_{A\subseteq\Omega} (-1)^{|A|+1}\left|\bigcap_{S_j\in A} S_j\right|.
\]

The potential of overlaps are proxies of the portions of overlaps $\pi(A)$. As we do not have non-trivial estimations of the lower bounds of $\pi(A)$s---$\pi(A)$ has a trivial lower bound of 0---we discard the terms in the sum for $\vert A \vert > 2$.
Truncating the expansion of the sum after the pairwise terms yields:
\begin{equation}\label{eq:bonferroni2}
|S|
\ge
\sum_{i=1}^n |S_i|
-
\sum_{1\le i<j\le n} |S_i\cap S_j|.
\end{equation}

This is equivalent to:
\begin{equation}\label{eq:population_size_PO}
    \vert \bigcup\limits_{i|S_i \in \Omega} S_i \vert \geq \sum\limits_{i|S_i \in \Omega} n_i - \sum\limits_{A\vert A \subseteq \Omega, \vert A \vert=2} (\frac{\pi(A)}{1+\pi(A)}\cdot\sum\limits_{j|S_j \in A} n_j)
\end{equation}
Proof of it can be found in Section~\ref{sec:p}.\\
By substitute the $\pi(A)$ on the RHS of the Equation~\ref{eq:population_size_PO} with $\tilde{\pi}'_{\mathcal P'}(A)$, and consider that the overlap of a pair of sample cannot be larger than the smaller sample of the two, we obtain following proxy for the lower bound of the meta-analysis sample size $\vert \bigcup\limits_{i|S_i \in \Omega} S_i \vert$:

\paragraph{Proxy of the lower bound of the meta-analysis sample size}
\begin{equation}
\widetilde{|S|}_{\mathrm{LB}}=\sum\limits_{i|S_i \in \Omega} n_i - \sum\limits_{A\vert A \subseteq \Omega, \vert A \vert=2} min((\frac{\tilde{\pi}'_{\mathcal P'}(A)}{1+\tilde{\pi}'_{\mathcal P'}(A)}\cdot\sum\limits_{j|S_j \in A} n_j), \min_{j|S_j \in A}n_j)
\end{equation}
The two $min$ operators make sure that the estimated pair-wise overlapping part do not exceed the sample size of the smaller sample of the pair.\\

\begin{remark}[Interpretation and limitation]\label{rmk:lb_interpretation}
The quantity $\widetilde{|S|}_{\mathrm{LB}}$ is deterministic given $n_i$ and the pairwise potential
$\tilde{\pi}'_{\mathcal P'}(\{S_i,S_j\})$.
Because $\tilde{\pi}'_{\mathcal P'}$ is a feasibility-based measure and is not guaranteed to satisfy
$\tilde{\pi}'_{\mathcal P'}(\{S_i,S_j\})\ge \pi(\{S_i,S_j\})$ in full generality, $\widetilde{|S|}_{\mathrm{LB}}$
should be interpreted as a conservative proxy for a lower bound rather than a universal guaranteed bound.
\end{remark}
\textbf{Example:}\\
From the 4-studies example, we list detailed calculation process below:\\
$$\widetilde{|S|}_{\mathrm{LB}}=\sum\limits_{i|S_i \in \Omega} n_i - \sum\limits_{A\vert A \subseteq \Omega, \vert A \vert=2} min((\frac{\tilde{\pi}'_{\mathcal P'}(A)}{1+\tilde{\pi}'_{\mathcal P'}(A)}\cdot\sum\limits_{j|S_j \in A} n_j), \min_{j|S_j \in A}n_j)$$
\begin{align*}
&=14
- \min\!\left\{\frac{\tfrac13}{1+\tfrac13}\cdot 7,\,3\right\}
- \min\!\left\{\frac{\tfrac13}{1+\tfrac13}\cdot 5,\,2\right\}
- \min\!\left\{0,\,3\right\} \\
&\quad
- \min\!\left\{\frac{\tfrac23}{1+\tfrac23}\cdot 6,\,2\right\}
- \min\!\left\{\frac{\tfrac14}{1+\tfrac14}\cdot 9,\,4\right\}
- \min\!\left\{\frac{\tfrac13}{1+\tfrac13}\cdot 7,\,2\right\} \\
&= 14
- \min\!\left\{\tfrac14\cdot 7,\,3\right\}
- \min\!\left\{\tfrac14\cdot 5,\,2\right\}
- 0
- \min\!\left\{\tfrac{2}{5}\cdot 6,\,2\right\}
- \min\!\left\{\tfrac{1}{5}\cdot 9,\,4\right\}
- \min\!\left\{\tfrac14\cdot 7,\,2\right\} \\
&= 14 - \min\{1.75,3\} - \min\{1.25,2\} - 0
    - \min\{2.4,2\} - \min\{1.8,4\} - \min\{1.75,2\} \\
&= 14 - 1.75 - 1.25 - 0 - 2 - 1.8 - 1.75 \\
&= 5.45.
\end{align*}

\begin{remark}\label{rmk:uopifpr}
In the formula for the estimation of the lower bound, only pair-wise potential of overlap were used. This enables a much faster calculation ($n^2$ vs. $2^n$), at the cost of information lost, which makes the lower bound over-conservative in the cases where there are overlap in combinations of 3 or more samples. We thus suggest to use it in combination with the 'Grid plot' when possible.
\end{remark}

\section{Case Studies}

\paragraph{Implementation steps (shared across cases).}
For each case study, we (i) selected key characteristics that were consistently reported across studies,
(ii) defined a reporting resolution by partitioning each global domain into bins,
(iii) encoded each study’s reported range as a set of bin indices, and
(iv) computed overlap potential $\tilde{\pi}'_{\mathcal P'}(A)$. 

For case 1 and case 2 we also derived overlap-free candidate sets.
Because $\tilde{\pi}'_{\mathcal P'}(A)$ depends on the chosen partitions, we report all partitions and the vectorization of them in respective figures, to ensure reproducibility. For readability, we put the large size figures in the appendix.
\paragraph{Notes on calculation efficiency}
In our implementation with R, we did not calculate the potential of overlap exactly one combination by one combination, as it grows exponentially with the number of studies($2^{51}$ with 51 studies!). Instead, we divide the whole calculation into multiple steps and prune the combinations('branches') with zero overlap potential at each step. Because this is not the central topic of this paper, we will not explain it in full detail.

\subsection{Case 1(a): Ethnicity and clinical outcomes in COVID-19}
Fifty-one studies were included in the systematic review. The key characteristics chosen here are the location of observation (locations in UK and US) and the setting of the hospital (community or Hospital). \cite{Sze2020}\\
The ranges of the key characteristics and the sample sizes were extracted and the results are shown in the Figure~\ref{fig: vectorization_51_studies}.\\
\begin{figure}[!htbp]
\includegraphics[width=1\textwidth]{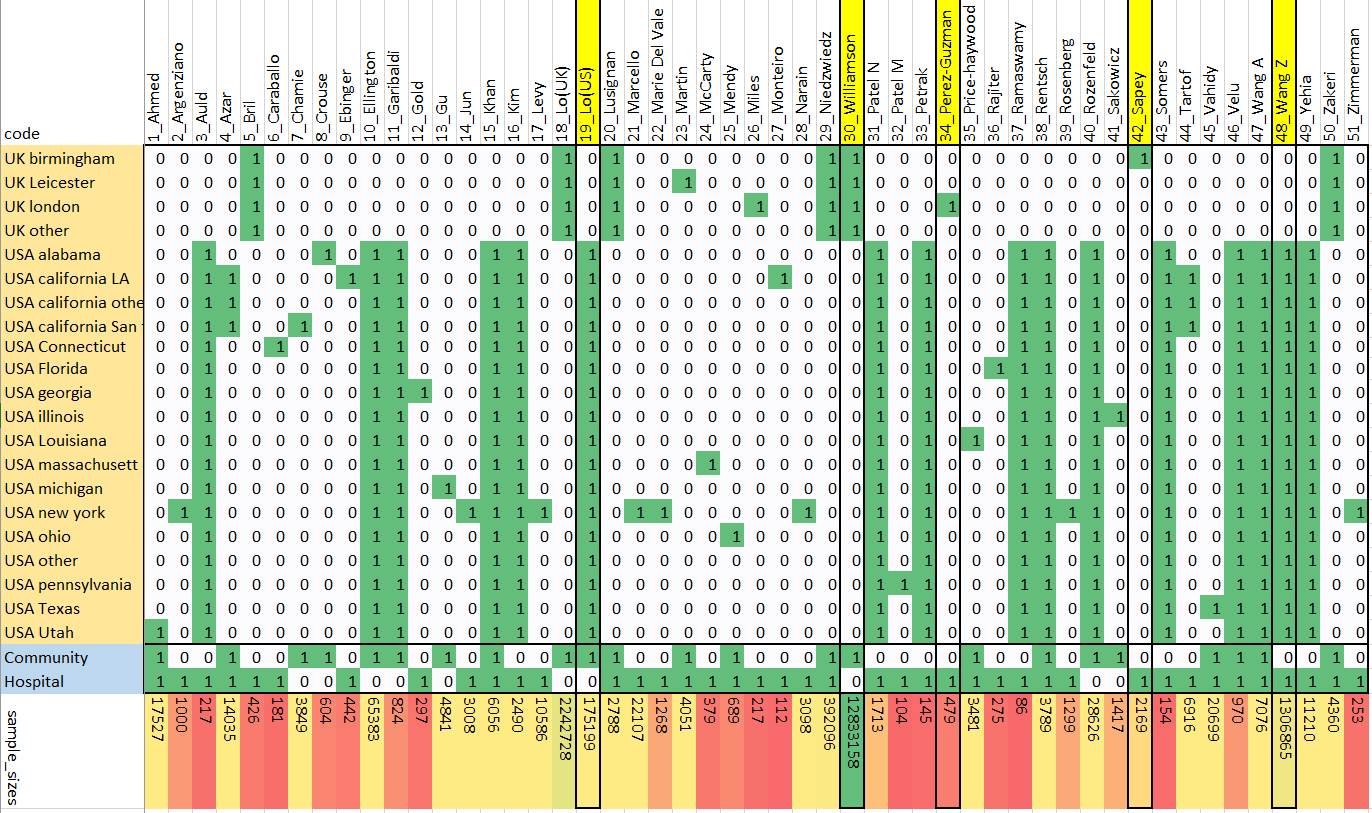}
    \caption{Case~1(a) (Sze et al.): Binary encoding (``vectorization'') of reported key-characteristic envelopes for 51 studies. For each study (column), rows correspond to bins of the chosen key characteristics (geographic location and hospital setting) at the specified reporting resolution; entries indicate whether the study's reported envelope intersects each bin. Studies highlighted in yellow belong to the selected overlap-free candidate set with the largest total sample size (Section~\ref{sec:ofs}).}

    \label{fig: vectorization_51_studies}
\end{figure}

We calculated the potential of overlaps for all combinations of study samples (as in Section~\ref{sec:roo}). At least 50 sets of studies have potential of overlaps of one, which means their ranges of characteristics are identical(at the chosen reporting resolution). In this case, manually checking the sample combinations one by one would be impractical. Instead, we used the algorithm introduced in Section~\ref{sec:ofs} to calculate the overlap-free sample combination with the largest sample size. The result is a set of 5 study samples $\{$19$\_$Lo(US), 30$\_$Williamson, 34$\_$Perez-Guzman, 42$\_$Sapey, 48$\_$Wang Z$\}$, which are marked in the Figure~\ref{fig: vectorization_51_studies}. Although we have only 5 studies after applying this algorithm, they have a total sample size of 14,317,870. Compared with the original total sample size of 17,211,742 from 51 studies, this means a reduction of $16.8\%$. In other words, we obtain an overlap-free meta-analysis sample, for which we can claim that the results of a systematic review on it is very unlikely to be biased (or underestimate the standard error) due to sample overlap, at the cost of a $16.8\%$ reduction in sample size.\\
The proxy of the lower bound of the meta-analysis sample size, defined in Section~\ref{sec:eub} is 13,482,612. In this case, the proxy is even lower than the 14,317,870 of the overlap-free sample combination with the largest sample size, which makes the proxy superfluous. This is partly due to the large amount of the study samples and the high degrees of overlap between the sample characteristics, which causes the pair-wise overlaps to be subtracted too many times.\\

\subsubsection{Case 1(b): Subset of case 1(a)}
The method can also be applied to any subset of the studies included in a meta-analysis. For example, if we perform the overlap-analysis on the first 10 studies (as were ordered alphabetically according to the name of the first author) instead of all 51, we obtain the following results: the summed sample size is 103,664; the graph of the potential of overlaps is shown in Figure~\ref{fig: 10studies_case}; the overlap-free sample combination with the largest sample size is $\{$2$\_$Argenziano, 5$\_$Bril, 6$\_$Caraballo, 9$\_$Ebinger, 10$\_$Ellington$\}$ with a sample size of 67,432; and the proxy of the lower bound of the meta-analysis sample size is 77,844(rounded). Compared to Case 1(a), this proxy is larger than the largest sample size among overlap-free sample combinations. In such cases, using the proxy for further analysis can potentially enable a higher statistical power in the results.\\
\begin{figure}[!htbp]
\includegraphics[width=0.8\textwidth]{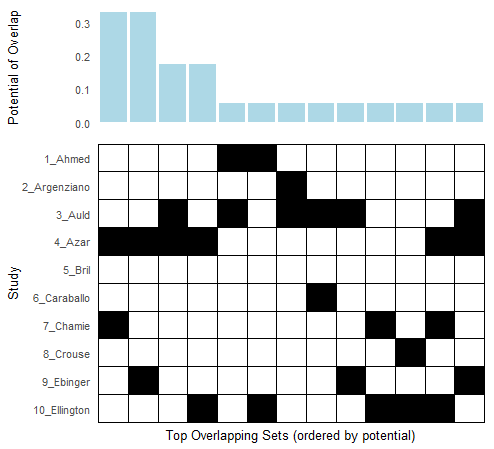}
    \caption{Case~1(b) (subset of Sze et al.): Visualization of overlap potentials for the first 10 included studies. The plot summarizes study-combinations with large values of $\tilde{\pi}'_{\mathcal P'}(A)$ under the chosen key characteristics and partition.}
    \label{fig: 10studies_case}
\end{figure}

\subsection{Case 2: Association of COVID-19 Pandemic Measures With Cancer Screening}
We analyzed the potential of overlap in the systematic review in the paper 'Global Association of COVID-19 Pandemic Measures With Cancer Screening A Systematic Review and Meta-analysis', \cite{Teglia2022}.

The paper investigated the pandemic’s association with cancer screening worldwide. Thirty-nine studies that reported data from cancer registries that compared the number of screening tests performed before and during the pandemic for breast, cervical, and colorectal cancer were included.

Considering the available information, we chose 'location', 'time' and 'disease' as the key characteristics. Then we extracted the ranges of the key characteristics for each study, selected appropriate partitions and performed vectorization. The results of the vectorization can be seen in the Figure~\ref{fig: vectorized_teglia}.

Looking at the sample combinations with the highest potential of overlap in Figure~\ref{fig: stepone_teglia_39_50}, we can see that study combination number 6 and number 17 have a potential of overlap of 1, which means their ranges of all key characteristics are identical(at the chosen reporting resolution). 

We then manually investigated the combinations with the highest potential of overlap and were able to find out that study 17 by Fillon et al. was a review of the findings presented in study 6 by McBain et al. This means the two articles actually used the identical samples.  

Such a discovery would require much more effort without our method, as checking pairwise relations alone would mean checking 210 non-zero combinations of studies. 

\subsection{Case 3: Variations in reported epidemiology of major lower extremity amputation}
We looked into another systematic review "Understanding variations in reported epidemiology of major lower extremity amputation in the UK: a systematic review" \cite{Meffen2021}. The overlap handling in this article was also investigated in \cite{Hussein2022}. Both the original authors and Hussein et al. reached the conclusion that the available information is insufficient to assess the degree of overlap. That makes it impossible to exclude articles based on overlap, and thus only a narrative summary was provided instead of a statistical data synthesis.

Our method provides a practical way to generate overlap-free candidate study sets (at the chosen reporting resolution),
which can then be used for sensitivity analyses or for constructing multiple non-overlapping meta-analyses (Section~\ref{sec:ofs}).

We selected the 'time of the observation' and the 'location of the observation' as the key characteristics and performed the vectorization, as shown in Figure~\ref{fig: Meffen_Vectorization}. The sample combinations with the highest potential of overlap are shown in Figure~\ref{fig: meffen_all_27}. We were able to calculate the set of overlap-free study combinations ($B_2$, as denoted in Section~\ref{sec:ofs}, comprising 298 combinations), which we included as supplementary material. From there, further steps can be taken to narrow down the study combinations to take, or we can perform multiple meta analyses, as suggested in Section~\ref{sec:ofs}.
\section{Conclusions and Discussion}
We proposed a model that adequately describes the overlap relationships among study samples in evidence synthesis. We introduced a quantity called the "potential of overlap", which indicates the degree of overlap among a given set of studies. We described the theory behind our method using set theory and the coding of basic information about the study samples. Our method is designed to be viable even when only very limited information about the sample is available, for example, when only the study papers are available, which is often the case in practice. 

As an example to demonstrate the practical value of the method, we designed an algorithm based on "potential of overlap" that finds out the overlap-free sets of samples with the largest sample size. This set can be used for further analysis, providing valuable insight of the potential impact of sample overlap. Similarly, based on this method, we showed it is also possible to estimate the lower bound of the overlap-free sample size of an evidence synthesis.

We applied our methods to existing evidence synthesis in this paper and were able to confirm the viability of our method in complex real-world scenarios. The results also demonstrate the necessity of such overlap analysis due to the high potential of overlaps we discovered using our method in the investigated cases.



Our methods can serve as a basis for further research on evaluating the impact of sample overlap as well as on the correction of evidence synthesis results regarding the overlap. More specifically, for meta-analysis with binary outcomes, we might be able to adjust the variance and the effect estimators. The lower bound of the overlap-free sample size can be used to adjust the estimators for sensitivity analysis in the context of best-case and worst-case scenarios. In that case, we might need to choose proper assumption on the sample characteristics to enable the calculation of a one-sided confidence interval for the proxy/estimator of the lower bound of the overlap-free sample size. If the precision of the proxy/estimator needs to be improved, we could consider different assumptions about the distribution of the characteristics or consider other function for the definition of the potential of overlap to adapt to different situations. In the long term, this research may improve the standardization and the transparency of overlap-handling, thereby further improving the credibility of evidence synthesis.

\section*{Statements and Declarations}
\begin{itemize}
\item Funding : No funding was received to assist with the preparation of this manuscript.
\item Conflict of interest/Competing interests (check journal-specific guidelines for which heading to use): The authors have no relevant financial or non-financial interests to disclose.
\item Ethics approval and consent to participate: Not relevant.
\item Consent for publication: 
\item Data availability: Codes that are used for the results are available upon request.
\item Materials availability: Not relevant.
\item Code availability: Codes that are used for the results are available upon request.
\item Author contribution:
Conceptualization: Tim Mathes, Zhentian Zhang\\
Methodology: Tim Mathes, Zhentian Zhang\\
Formal analysis, investigation: Zhentian Zhang\\
Writing - original draft preparation: Zhentian Zhang\\
Writing - review and editing: Tim Mathes, Zhentian Zhang\\
Funding acquisition: 
Supervision: Tim Mathes, Tim Friede

\end{itemize}
\section{MSC code}
We suggest 62P10-Applications of statistics to biology and medical sciences; meta analysis as the classification code.

\bibliography{sn-bibliography}

@article{Bom2020,
  author  = "Bom, P.R.D. and Rachinger, H.",
  title   = "A generalized-weights solution to sample overlap in meta-analysis",
  year    = "2020",
  journal = "Research Synthesis Methods",
  volume  = "11",
  pages   = "812--832",
  doi     = "10.1002/jrsm.1441",
  url     = "https://doi.org/10.1002/jrsm.1441"
}

@article{Han2016,
  author  = "Han, B. and others",
  title   = "A general framework for meta-analyzing dependent studies with overlapping subjects in association mapping",
  year    = "2016",
  journal = "Human Molecular Genetics",
  volume  = "25",
  pages   = "1857--1866",
  doi     = "10.1093/hmg/ddw049",
  note    = "Epub 2016 Feb 21"
}

@article{Hussein2022,
  author  = "Hussein, H. and others",
  title   = "Double-counting of populations in evidence synthesis in public health: a call for awareness and future methodological development",
  year    = "2022",
  journal = "BMC Public Health",
  volume  = "22",
  pages   = "1827",
  doi     = "10.1186/s12889-022-14213-6"
}

@article{Jin2020,
  author  = "Jin, Q. and Shi, G.",
  title   = "Meta-Analysis of SNP-Environment Interaction With Overlapping Data",
  year    = "2020",
  journal = "Frontiers in Genetics",
  volume  = "10",
  pages   = "1400",
  doi     = "10.3389/fgene.2019.01400",
  url     = "https://doi.org/10.3389/fgene.2019.01400"
}

@article{Lin2009,
  author  = "Lin, D.-Y. and Sullivan, P.F.",
  title   = "Meta-Analysis of Genome-wide Association Studies with Overlapping Subjects",
  year    = "2009",
  journal = "American Journal of Human Genetics",
  volume  = "85",
  pages   = "862--872",
  doi     = "10.1016/j.ajhg.2009.11.001"
}

@article{Lunny2021,
  author  = "Lunny, C. and Pieper, D. and Thabet, P. and Kanji, S.",
  title   = "Managing overlap of primary study results across systematic reviews: practical considerations for authors of overviews of reviews",
  year    = "2021",
  journal = "BMC Medical Research Methodology",
  volume  = "21",
  pages   = "140",
  doi     = "10.1186/s12874-021-01269-y",
  url     = "https://doi.org/10.1186/s12874-021-01269-y"
}

@article{Mathes2023,
  author  = "Mathes, T. and others",
  title   = "Systematic reviews and meta-analyses that include registry-based studies: methodological challenges and areas for future research",
  year    = "2023",
  journal = "Journal of Clinical Epidemiology",
  volume  = "156",
  pages   = "119--122",
  doi     = "10.1016/j.jclinepi.2023.02.014",
  note    = "Epub 2023 Feb 16"
}

@article{Meffen2021,
  author  = "Meffen, A. and Houghton, J.S.M. and Nickinson, A.T.O. and others",
  title   = "Understanding variations in reported epidemiology of major lower extremity amputation in the UK: a systematic review",
  year    = "2021",
  journal = "BMJ Open",
  volume  = "11",
  pages   = "e053599",
  doi     = "10.1136/bmjopen-2021-053599",
  url     = "http://dx.doi.org/10.1136/bmjopen-2021-053599"
}

@article{Sze2020,
  author  = "Sze, S. and Pan, D. and Nevill, C.R. and Gray, L.J. and Martin, C.A. and Nazareth, J. and Minhas, J.S. and Divall, P. and Khunti, K. and Abrams, K.R. and Nellums, L.B. and Pareek, M.",
  title   = "Ethnicity and clinical outcomes in COVID-19: A systematic review and meta-analysis",
  year    = "2020",
  journal = "EClinicalMedicine",
  volume  = "29--30",
  pages   = "100630",
  doi     = "10.1016/j.eclinm.2020.100630"
}

@article{Teglia2022,
  author  = "Teglia, F. and Angelini, M. and Astolfi, L. and Casolari, G. and Boffetta, P.",
  title   = "Global Association of COVID-19 Pandemic Measures With Cancer Screening: A Systematic Review and Meta-analysis",
  year    = "2022",
  journal = "JAMA Oncology",
  volume  = "8",
  number  = "9",
  pages   = "1287--1293",
  doi     = "10.1001/jamaoncol.2022.2617"
}

@article{Wolery2010,
  author  = "Wolery, M. and Busick, M. and Reichow, B. and Barton, E.E.",
  title   = "Comparison of Overlap Methods for Quantitatively Synthesizing Single-Subject Data",
  year    = "2010",
  journal = "Journal of Special Education",
  volume  = "44",
  pages   = "18--28",
  doi     = "10.1177/0022466908328009",
  url     = "https://doi.org/10.1177/0022466908328009"
}

@book{Gliklich2020Registry,
  editor      = {Gliklich, Richard E. and Leavy, Michelle B. and Dreyer, Nancy A.},
  title       = {Registries for Evaluating Patient Outcomes: A User's Guide},
  edition     = {4th},
  year        = {2020},
  month       = sep,
  address     = {Rockville (MD)},
  publisher   = {Agency for Healthcare Research and Quality (US)},
  series      = {AHRQ Methods for Effective Health Care},
  note        = {Report No.: 19(20)-EHC020. Bookshelf ID: NBK562575},
  url         = {https://www.ncbi.nlm.nih.gov/books/NBK562575/}
}

@article{Lex2014UpSet,
  author  = {Lex, Alexander and Gehlenborg, Nils and Strobelt, Hendrik and Vuillemot, Romain and Pfister, Hanspeter},
  title   = {UpSet: Visualization of Intersecting Sets},
  year    = {2014},
  journal = {IEEE Transactions on Visualization and Computer Graphics},
  volume  = {20},
  number  = {12},
  pages   = {1983--1992},
  doi     = {10.1109/TVCG.2014.2346248}
}

\backmatter

\begin{appendices}

\section{Proofs}\label{sec:p}
\textbf{Proof of Proposition}
\begin{proof}

\begin{align*}
     &\exists k, \bigcap\limits_{i\vert S_i\in A} R_{i,k} = \emptyset\\
     &\Rightarrow \exists k, \bigcap\limits_{i\vert S_i\in A} D_{i,k} = \emptyset\\
     &\Rightarrow \bigcap\limits_{i\vert S_i\in A} S_i = \emptyset
\end{align*}
\end{proof}
\textbf{Proof of proposition~\ref{prop:epobotrvokc}}

\begin{proof}

\begin{align*}
     \mathbf{r}_{i_1,k} \cdot \mathbf{r}_{i_2,k} = 0 &\Rightarrow  \forall l\in \{1, 2,..., m_k\}, r_{i_1,k,l} \cdot r_{i_2,k,l} = 0\\
     &\Rightarrow \forall l\in \{1, 2,..., m_k\}, R_{i_1,k}\cap R_{\cdot,k,l} \cap R_{i_2,k}= \emptyset\\
    &\Rightarrow R_{i_1,k}\cap R_{\cdot,k} \cap R_{i_2,k}= \emptyset\\
    &\Rightarrow R_{i_1,k}\cap R_{i_2,k}= \emptyset\\
    &\Rightarrow D_{i_1,k}\cap D_{i_2,k}= \emptyset\\
    &\Rightarrow S_{i_1} \cap S_{i_2} = \emptyset
\end{align*}
\end{proof}
\textbf{Proof of equation~\ref{eq:population_size_PO}}

\begin{proof}

\begin{align*}
     \vert \bigcup\limits_{i|S_i \in \Omega} S_i \vert &=\sum\limits_{A \subseteq \Omega}(-1)^{\vert A \vert +1}\vert \bigcap\limits_{i|S_i \in A} S_i \vert \\
     &\geq \sum\limits_{A\vert A \subseteq \Omega, \vert A\vert=1}\vert \bigcap\limits_{i|S_i \in A} S_i \vert-\sum\limits_{A\vert A \subseteq \Omega, \vert A\vert=2}\vert \bigcap\limits_{j|S_j \in A} S_j \vert\\
    &= \sum\limits_{i|S_i \in \Omega} n_i-\sum\limits_{A\vert A \subseteq \Omega, \vert A\vert=2}\vert \bigcap\limits_{j|S_j \in A} S_j \vert\\
    &= \sum\limits_{i|S_i \in \Omega} n_i-\sum\limits_{A\vert A \subseteq \Omega, \vert A\vert=2} (\frac{ \vert \bigcap\limits_{j\vert S_j\in A} S_j \vert}{ \vert \bigcup\limits_{j\vert S_j\in A} S_j \vert}\cdot\vert\bigcup\limits_{i|S_i \in A} S_i \vert)\\
    &= \sum\limits_{i|S_i \in \Omega} n_i-\sum\limits_{A\vert A \subseteq \Omega, \vert A\vert=2} (\pi(A)\cdot\vert\bigcup\limits_{j|S_j \in A} S_j \vert)\\
    &= \sum\limits_{i|S_i \in \Omega} n_i-\sum\limits_{A\vert A \subseteq \Omega, \vert A\vert=2} (\frac{\pi(A)}{1+\pi(A)}\cdot\sum\limits_{j|S_j \in A} n_j)\\
\end{align*}
\end{proof}

\section{Aggregated data: Study-level (aggregate) information extracted from reports}\label{app:aggdata}

Throughout the paper we use the term \emph{aggregated data} (also called study-level information)
to mean information that can be obtained from a study report, protocol, supplement, or registry entry
without access to individual participant data (IPD).
This includes both (i) classical aggregate/summary statistics used in standard aggregate-data meta-analysis,
and (ii) eligibility envelopes (ranges/sets) of key characteristics that define which observations
could have been included in the study.

\paragraph{What we treat as aggregated data (used in this paper).}
For each included study $i$, the aggregated information potentially consists of:
\begin{itemize}
\item \textbf{Sample size information:} total sample size $n_i$, and when applicable subgroup/arm sizes
(e.g., treatment/control counts, number of screened tests, number of events).
\item \textbf{Effect summary information (for downstream synthesis):} reported effect estimates and measures
of uncertainty (e.g., $\hat\theta_i$ with standard error, confidence interval, $p$-value), or contingency tables
(e.g., $2\times 2$ tables) from which these can be derived.
\item \textbf{Eligibility envelopes for key characteristics (central to overlap inference here):}
reported restrictions that define the study sample, represented as sets $R'_{i,k}$.
Typical examples are:
\begin{itemize}
\item time windows (e.g., years 2010--2019, months, or specific periods),
\item geographic regions (e.g., countries, states, hospital catchment areas),
\item setting definitions (e.g., hospital vs community),
\item diagnosis/case definitions (e.g., ICD code lists at a specified coding level),
\item age ranges or other baseline restrictions (e.g., 18--65, sex).
\end{itemize}
The envelope $R'_{i,k}$ may be an interval, a finite set, a union of intervals, or a categorical set.
\end{itemize}

\paragraph{What we do not mean by aggregated data.}
\begin{itemize}
\item \textbf{Not IPD:} we do not observe the individual values $d'_{I_{i,j},k}$, unique identifiers,
or any cross-study linkage of records.
\item \textbf{Not study-specific administrative metadata for overlap inference:}
bibliographic or administrative fields such as author list, publication year, registry identifier,
data-extraction timestamp, or insurance status are not intrinsic to an observation event and are therefore
not treated as key characteristics for overlap inference (cf.\ Section~\ref{sec:example_real_world}).
We may record such metadata for documentation, but it does not enter $\tilde{\pi}'_{\mathcal P'}$.
\end{itemize}

\paragraph{Why we include eligibility envelopes as aggregated data.}
In many observational studies, overlap risk is driven by shared data sources and overlapping inclusion criteria.
Even when effect summaries are available, overlap cannot be assessed from effect estimates alone.
Our approach therefore treats reported eligibility restrictions on key characteristics as part of the study-level
information, and uses them to compute the overlap potential $\tilde{\pi}'_{\mathcal P'}(A)$.

\section{Large size Graphs}\label{sec:G}
\begin{figure}[ht]
    \centering
    \includegraphics[width=1\textwidth]{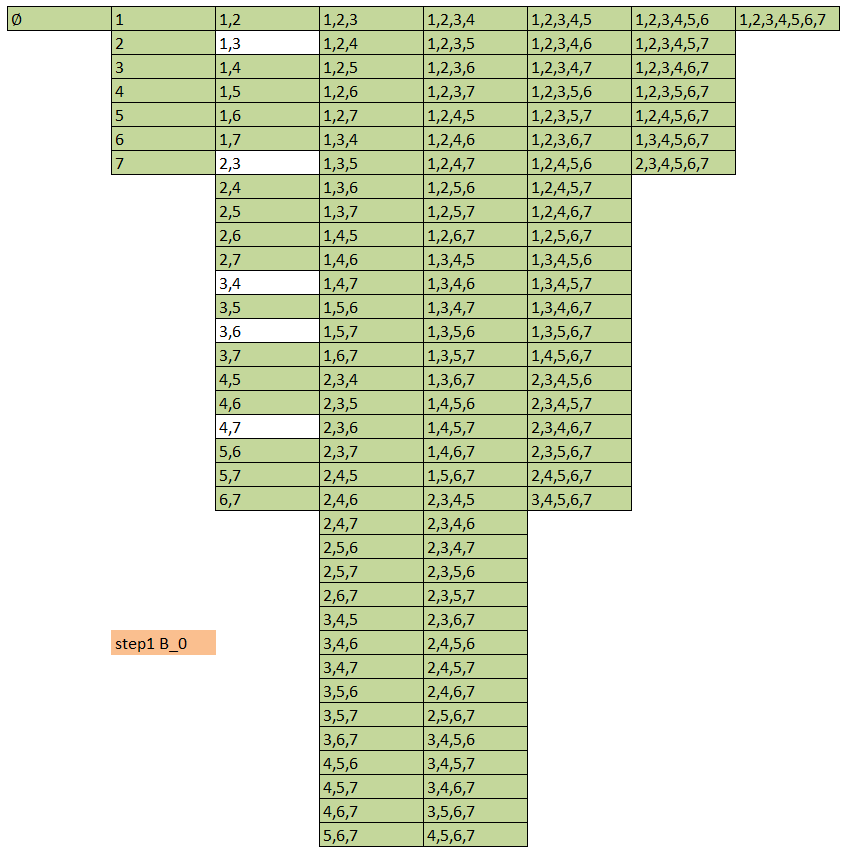}
    \caption{Example: Overlap-free subsets $B_0=\{A\subseteq\Omega:\tilde{\pi}'_{\mathcal P'}(A)=0\}$. Each cell represents a study-combination $A$, with filled cells indicating subsets for which overlap is excluded at the chosen reporting resolution.}
    \label{fig:7studyb0}
\end{figure}

\begin{figure}[ht]
    \centering
    \includegraphics[width=1\textwidth]{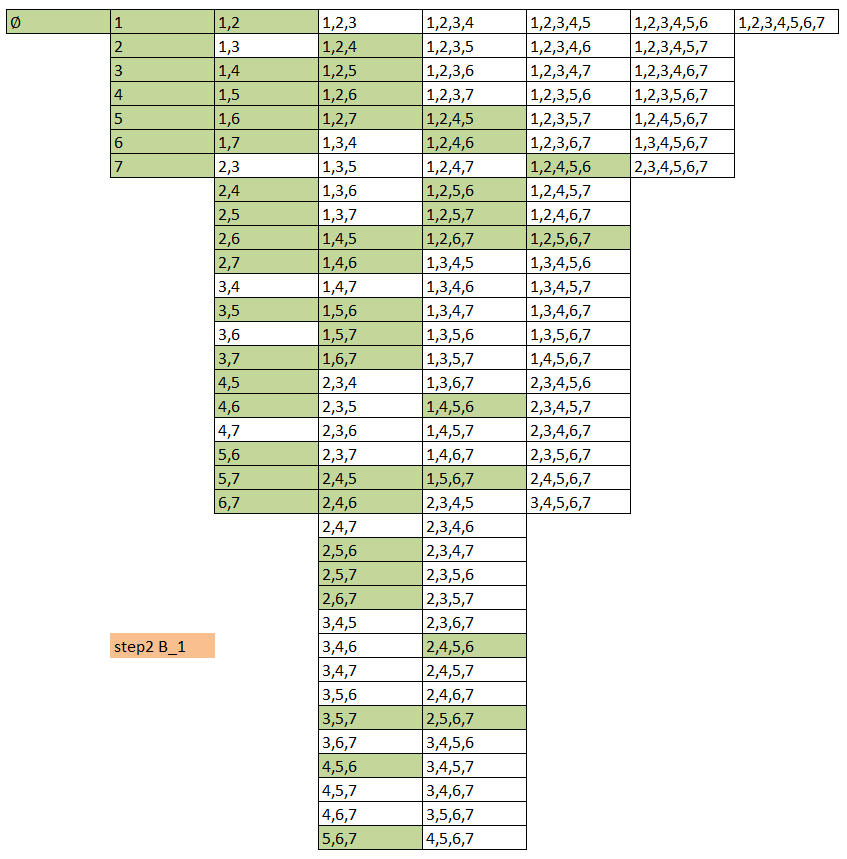}
    \caption{Example: Overlap-free subsets $B_1\subseteq B_0$, retaining only those $A\in B_0$ for which every subset of $A$ also lies in $B_0$.}

    \label{fig:7studyb1}
\end{figure}

\begin{figure}[ht]
    \centering
    \includegraphics[width=1\textwidth]{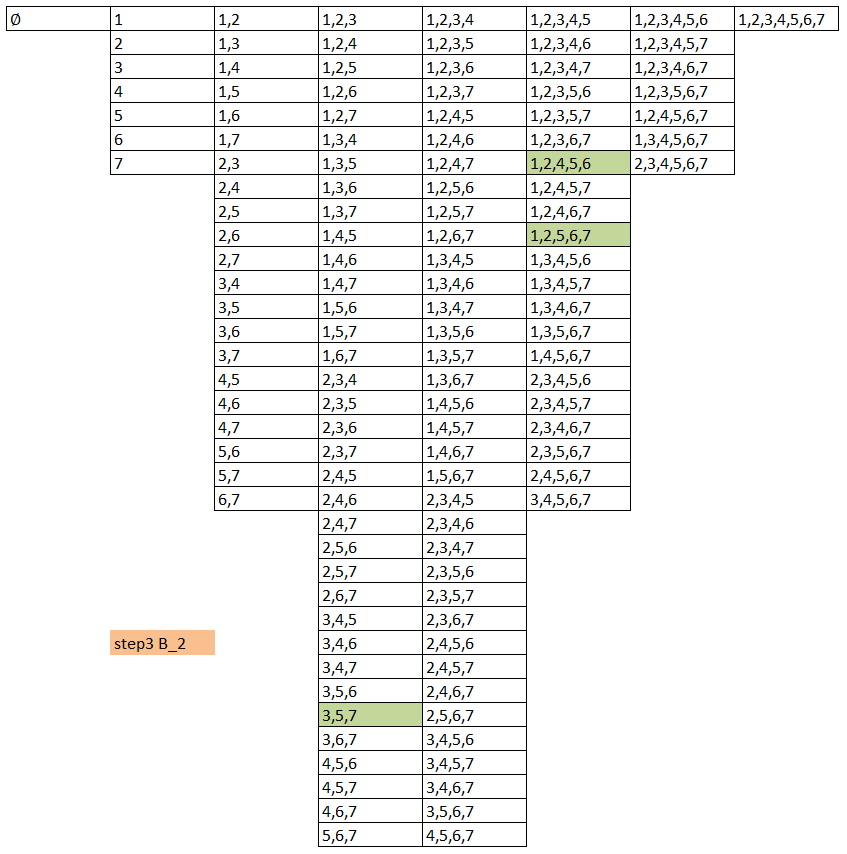}
    \caption{Example: Overlap-free subsets $B_2\subseteq B_1$, retaining only those $A\in B_1$ that have no strict superset in $B_1$. These maximal sets serve as overlap-free candidate study sets for subsequent selection criteria (Section~\ref{sec:ofs}).}
    \label{fig:7studyb2}
\end{figure}

\begin{figure}[!htbp]
\includegraphics[width=1\textwidth]{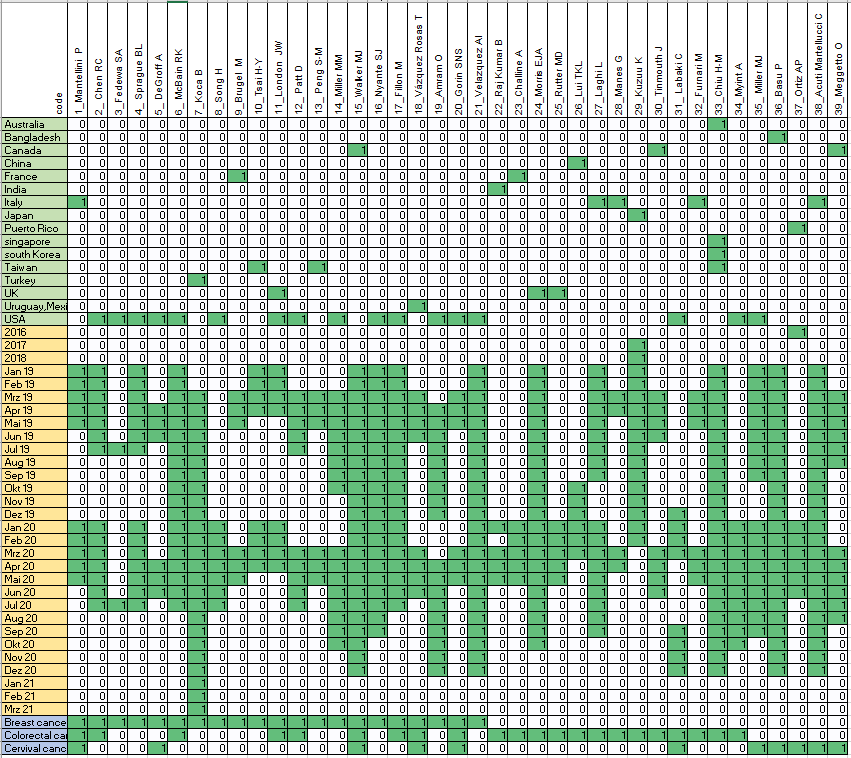}
    \caption{Case 2: Vectorization, Teglia et al.}
    \label{fig: vectorized_teglia}
\end{figure}

\begin{figure}[!htbp]
\includegraphics[width=1\textwidth]{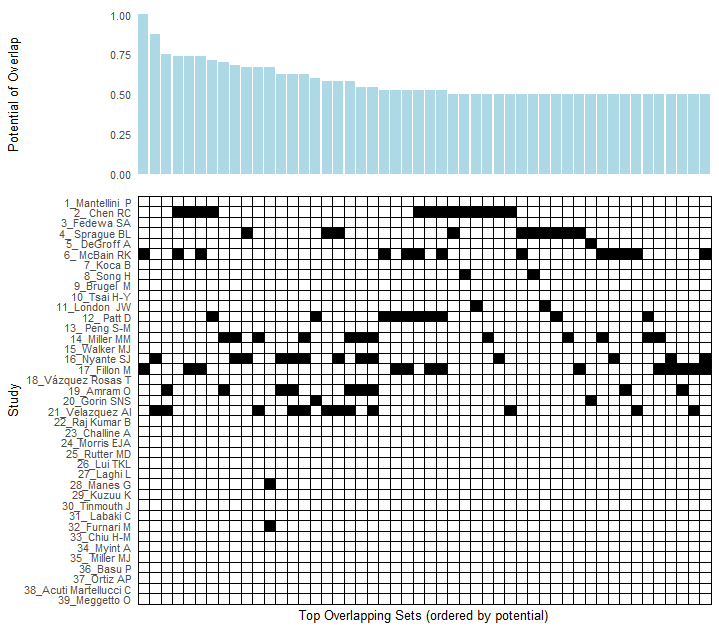}
    \caption{Case~2 (Teglia et al.): Study-combinations with the 50 largest overlap potentials. Each column corresponds to a subset $A\subseteq\Omega$ (membership indicated by filled cells); the value shown for each column is $\tilde{\pi}'_{\mathcal P'}(A)$. Columns are ordered by decreasing overlap potential. Combinations with value $1$ have identical binned envelopes across all selected key characteristics at the chosen resolution.}
    \label{fig: stepone_teglia_39_50}
\end{figure}

\begin{figure}[!htbp]
\includegraphics[width=0.60\textwidth]{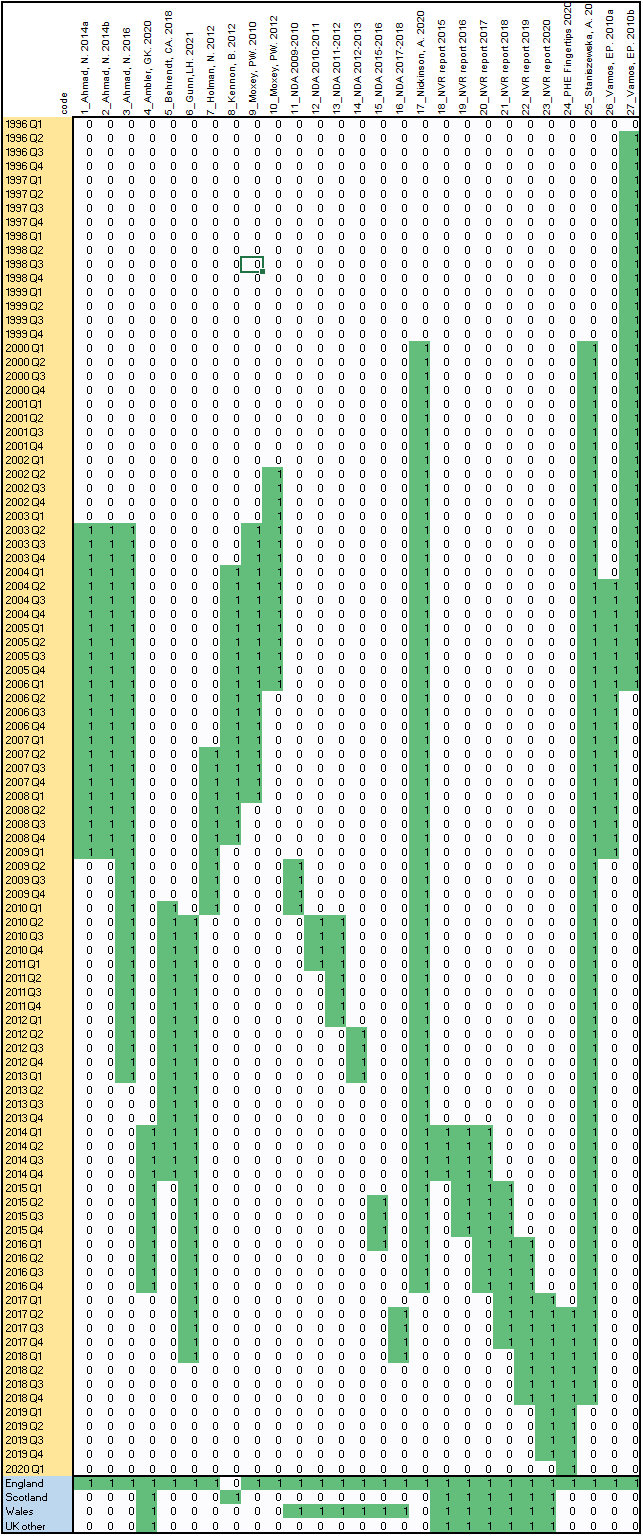}
    \caption{Case~3 (Meffen et al.): Binary encoding of reported envelopes for the key characteristics (time window and region) at the chosen reporting resolution. Columns correspond to studies; rows correspond to bins of each key characteristic. This encoding is used to compute overlap potentials and to enumerate overlap-free candidate study sets.}
    \label{fig: Meffen_Vectorization}
\end{figure}

\begin{figure}[!htbp]
\includegraphics[width=0.8\textwidth]{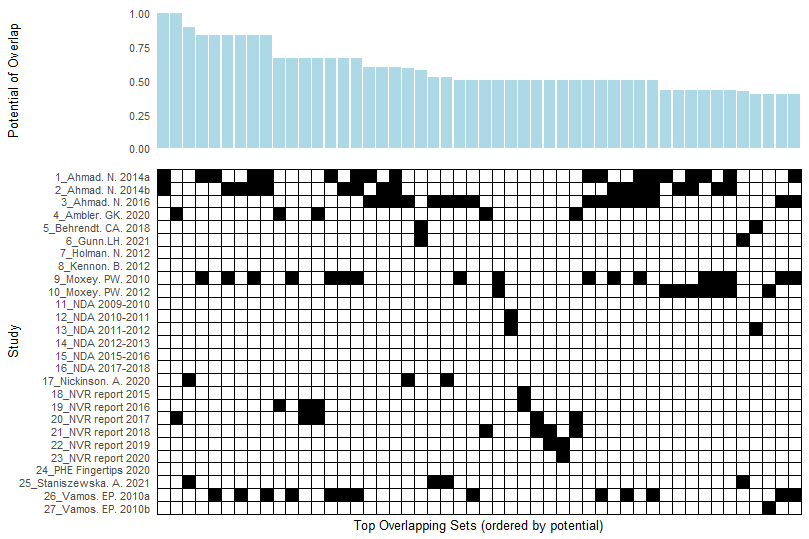}
    \caption{Case~3 (Meffen et al.): Study-combinations with the 50 largest overlap potentials under the chosen key characteristics and partition. Each column denotes a subset $A\subseteq\Omega$ (membership indicated by filled cells); values report $\tilde{\pi}'_{\mathcal P'}(A)$ under Assumption~\ref{ass:partition_compatibility}.}
    \label{fig: meffen_all_27}
\end{figure}




\end{appendices}


\end{document}